%% file: IntConcExp-2025-final.tex
\documentclass[submission,copyright,creativecommons]{eptcs}
\providecommand{\event}{ICE 2025} % Name of the event you are submitting to

\usepackage{iftex}

\usepackage[T1]{fontenc}
\usepackage{graphicx}
\usepackage{amsmath,amssymb,amsthm}
\usepackage{bussproofs}
\usepackage{breakurl}
\usepackage{calc}
\usepackage{xcolor,graphicx}
\usepackage{enumitem}
\usepackage{mathtools}
\usepackage{url}

\ifpdf
  \usepackage{underscore}         % Only needed if you use pdflatex.
  \usepackage[T1]{fontenc}        % Recommended with pdflatex
\else
  \usepackage{breakurl}           % Not needed if you use pdflatex only.
\fi

\input{macros.tex}
  \numberwithin{equation}{section}
\usepackage{tikz,fp,tikz-cd,tkz-graph}
\usepackage{pgfplots}
\usetikzlibrary{arrows.meta,arrows,calc,automata,angles,quotes,matrix,%backgrounds,calc,chains,fadings,fit,patterns,snakes,
                positioning,intersections,through,shapes.geometric,shapes.symbols,shapes.multipart,through,trees,
                decorations.markings,decorations.text,
                } 
\title{Bisimilarity and Simulatability of Processes\\
       Parameterized by Join Interactions}
\author{Clemens Grabmayer
\institute{GSSI\\ L'Aquila, Italy}
\email{clemens.grabmayer@gssi.it}
\and
Maurizio Murgia 
\institute{GSSI\\ L'Aquila, Italy}
\email{\quad maurizio.murgia@gssi.it}
}
\def\titlerunning{Bisimilarity and Simulatability of Processes Parameterized by Join Interactions}
\def\authorrunning{C.\ Grabmayer \& M.\ Murgia}
\begin{document}
\maketitle

%-------------------
% abstract
%-------------------
\input{abstract.tex}

%-------------------

%------------------------------
\section{\protect\label{intro}%
         Introduction}
%------------------------------

\input{intro.tex}%
  %\label{intro}
%------------------------------

%\newpage
%-------------------------------------------------------------
\section{\protect\label{prelims}%
         Preliminaries on Larsen's parameterized bisimilarity}
%-------------------------------------------------------------
\input{prelims.tex}%
  %\label{prelims}
%-------------------------------------------------------------

%\newpage
%-----------------------------------------------------------------------
\section{\protect\label{ji-param}%
         Join-Interaction parameterized simulatability and bisimilarity}
%-----------------------------------------------------------------------
\input{ji-param.tex}%
  %\label{ji-param}
%-------------------------------------------------------------------

%  \newpage
%-----------------------------------------------------------------------------
\section{\protect\label{log-char}%
         Modal characterization of \protect\jiorparameterized\ simulatability}
%-----------------------------------------------------------------------------
\input{log-char.tex}%
  %\label{log-char}
%-------------------------------------------------------------

%-------------------
\section{\protect\label{concl}%
         Conclusion (summary, literature, current work, open problems, plans)}
%-------------------
\input{concl.tex}%
  %\label{concl}
%-------------------

%--------------------------------------- 
% bibliography
%--------------------------------------- 
\nocite{*}
\bibliographystyle{eptcs}
\bibliography{IntConcExp-2025-final.bib}
%---------------------------------------

% \newpage
% %------------------
% \section*{Appendix}
% %------------------
% \input{appendix}
% %------------------

\end{document}

%% file: macros.tex
\definecolor{azure}{rgb}{0.94,1.00,1.00}
\definecolor{brown}{rgb}{.75,.25,.25}
\definecolor{cyan}{rgb}{0.25,0.88,0.82}
\definecolor{chocolate}{rgb}{0.82,0.41,0.12}
\definecolor{darkcyan}{rgb}{0.5,0,1}
\definecolor{darkgreen}{rgb}{0,0.39,0}
\definecolor{darkmagenta}{rgb}{0.5,0,0.5}
\definecolor{darkgoldenrod}{RGB}{184,134,11}
\definecolor{firebrick}{RGB}{175,25,25}
\definecolor{forestgreen}{rgb}{0.13,0.55,0.13}
\definecolor{goldenrod}{RGB}{218,165,32}
\definecolor{grey}{RGB}{128,128,128}
\definecolor{lightcyan}{rgb}{0.88,1.00,1.00}
\definecolor{lightpink}{rgb}{1.00,0.71,0.76}
\definecolor{myyellow}{RGB}{235,235,0}
\definecolor{lightyellow}{rgb}{1.00,1.00,0.88}
\definecolor{lightgoldenrod}{rgb}{0.83,0.97,0.51}
\definecolor{lightgoldenrodyellow}{rgb}{0.98,0.98,0.82}
\definecolor{lightskyblue}{rgb}{0.53,0.81,0.98}
\definecolor{moccasin}{rgb}{1.00,0.89,0.71}
\definecolor{magenta}{rgb}{1,0,1}
\definecolor{navyblue}{rgb}{0,0,0.5}
\definecolor{orange}{rgb}{1.0,0.65,0.0}
\definecolor{orangered}{rgb}{1.0,0.27,0.0}
\definecolor{palegreen}{rgb}{0.60,0.98,0.60}
\definecolor{powderblue}{rgb}{0.69,0.88,0.90}
\definecolor{purple}{rgb}{1,0.5,1}
\definecolor{royalblue}{RGB}{65,105,225}
\definecolor{mediumblue}{RGB}{0,0,205}
\definecolor{cornflowerblue}{RGB}{100,149,237}
\definecolor{springgreen}{rgb}{0.0,1.0,0.5}
\definecolor{turquoise}{rgb}{0.25,0.88,0.82}
\definecolor{snow}{rgb}{1.00,0.98,0.98}
\definecolor{tan}{rgb}{0.82,0.71,0.55}
\definecolor{red}{rgb}{1,0,0}
\definecolor{violetred}{RGB}{208,32,144}

\newcommand{\colorin}[1]{\textcolor{#1}}

\newcommand{\magenta}{\colorin{magenta}}

\newenvironment{alert}{\color{red}}{\color{black}}
\newenvironment{new}{\color{chocolate}}{\color{black}}
\newenvironment{newer}{\color{firebrick}}{\color{black}}
\newenvironment{newest}{\color{red}}{\color{black}}
\newenvironment{revised}{\color{violetred}}{\color{black}}

\newcommand{\isnewest}[1]{\begin{newest}{#1}\end{newest}}

\renewcommand{\alert}{\isnewest}
\newcommand{\nb}{\nobreakdash}

\newcommand{\nf}{\normalfont}

\newcommand{\sdefdby}{{:=}}
\newcommand{\defdby}{\mathrel{\sdefdby}}

\newcommand{\funin}{\mathrel{:}}
\newcommand{\fap}[2]{{#1}(\hspace*{-0.5pt}{#2}\hspace*{-0.5pt})}

\newcommand{\iap}[2]{#1_{#2}}
\newcommand{\bap}[2]{#1_{#2}}
\newcommand{\pap}[2]{#1^{#2}}

\newcommand{\pbap}[3]{#1_{#3}^{#2}}

\newcommand{\family}[2]{\setexp{#1}_{#2}}
\newcommand{\familynormalsize}[2]{\setexpnormalsize{#1}_{#2}}
\newcommand{\familyns}{\familynormalsize}

\newcommand{\setexp}[1]{\left\{{#1}\right\}}
\newcommand{\setexpnormalsize}[1]{\{{#1}\}}

\newcommand{\descsetexpmid}{\mathrel{\vert}}
\newcommand{\descsetexpbigmid}{\mathrel{\big\vert}}

\newcommand{\descsetexp}[2]{\left\{{#1}\descsetexpmid{#2}\right\}}

\newcommand{\descsetexpns}[2]{\{{#1}\descsetexpmid{#2}\}}
\newcommand{\descsetexpbig}[2]{\bigl\{{#1}\descsetexpbigmid{#2}\bigr\}}

\newcommand{\cardinalityns}[1]{|#1|}

\newcommand\tuple[1]{\langle #1 \rangle}

\newcommand\tuplespace{\hspace*{0.5pt}}
\newcommand\pair[2]{\tuple{#1, \tuplespace #2}}

\newcommand{\True}{\top}

\newcommand{\slogand}{\wedge}

\newcommand{\logand}{\mathrel{\slogand}}
 
\newcommand{\slognot}{\neg}
\newcommand{\lognot}[1]{\slognot{#1}}
\newcommand{\existsstzero}[1]{\exists{\hspace*{1pt}#1}}

\newcommand{\forallstzero}[1]{\forall{\hspace*{0.5pt}#1}}

\newcommand{\diamondact}[1]{\langle{#1}\rangle}

\newcommand{\aform}{\phi}
\newcommand{\bform}{\psi}
\newcommand{\aformi}{\iap{\phi}}

\newcommand{\aclassforms}{\mathcal{F}}

\newcommand{\ssatisfies}{{\vDash}}
\newcommand{\satisfies}{\vDash}
\newcommand{\notsatisfies}{\not\hspace*{1pt}\vDash}

\newcommand{\forms}{\mathcal{M}}
\newcommand{\formsof}{\fap{\forms}}
\newcommand{\posforms}{\mathcal{L}}
\newcommand{\posformsof}{\fap{\posforms}}

\newcommand{\someforms}{\mathcal{F}}
\newcommand{\someformsof}{\fap{\someforms}}

\newcommand{\negclosure}[1]{\prescript{\neg}{}{\overline{#1}}}
\newcommand{\posform}[1]{{\left|{#1}\right|_{+}}} 
\newcommand{\actprojform}[1]{{\left|{#1}\right|_{\iap{\pi}{\actions}}}} 

\newcommand{\BNFor}{\:\mid\:}
\newcommand{\BNFdefdby}{\:{::=}\:}

\theoremstyle{plain}
\newtheorem{thm}{Theorem}[section]

\newtheorem{lem}[thm]{Lemma}

\newtheorem{prop}[thm]{Proposition}

\theoremstyle{definition}
\newtheorem{defi}[thm]{Definition}

\newcommand{\derivativeact}[1]{{$#1$}\nb-de\-riv\-a\-tive}
\newcommand{\derivativesact}[1]{\derivativeact{#1}s}

\newcommand{\imagefinite}{im\-age-fi\-nite}
\newcommand{\indexedby}[1]{{$#1$}\nb-in\-dexed}
\newcommand{\jiparameterized}{ji\nb-pa\-ram\-e\-ter\-ized}
\newcommand{\joininteraction}{join-in\-ter\-ac\-tion}

\newcommand{\jiorparameterized}{(ji\protect\nb-)pa\-ram\-e\-ter\-ized}

\newcommand{\labeledi}[1]{{$#1$}\nb-la\-beled}

\newcommand{\modallogical}{mod\-al-log\-ic\-al}

\newcommand{\nonempty}{non-emp\-ty}

\newcommand{\parameterizedby}[1]{{$#1$}\nb-pa\-ram\-e\-ter\-ized}

\newcommand{\preorder}{pre\-or\-der}
\newcommand{\preorders}{\preorder{s}}

\newcommand{\transitionact}[1]{{${#1}$}\nb-tran\-si\-tion}
\newcommand{\transitionsact}[1]{\transitionact{#1}s}

\newcommand{\aact}{a}
\newcommand{\bact}{b}

\newcommand{\actions}{A}

\newcommand{\aenv}{e}
\newcommand{\benv}{f}
\newcommand{\cenv}{g}
\newcommand{\aenvacc}{{\aenv'}}
\newcommand{\benvacc}{{\benv'}}

\newcommand{\aenvdacc}{{\aenv''}}

\newcommand{\aenvacci}[1]{{\aenv'_{#1}}}

\newcommand{\aenvi}{\iap{\aenv}}

\newcommand{\LTS}{LTS}
\newcommand{\LTSs}{LTSs} 

\newcommand{\aLTS}{\mathcal{T}}
\newcommand{\aLTSi}{\bap{\aLTS}}

\newcommand{\states}{\textrm{St}}

\newcommand{\astate}{s}
\newcommand{\astatei}{\bap{\astate}}
\newcommand{\astateacc}{\pap{\astate}{\prime}}
\newcommand{\astateacci}{\pbap{\astate}{\prime}}
\newcommand{\bstate}{t}

\newcommand{\bstateacc}{\pap{\bstate}{\prime}}

\newcommand{\aprocLTS}{\mathcal{P}}
\newcommand{\aenvLTS}{\mathcal{E}}

\newcommand{\procs}{\textrm{Pr}}
\newcommand{\procsi}{\bap{\procs}}

\newcommand{\envs}{\textrm{Env}}

\newcommand{\aproc}{p}
\newcommand{\bproc}{q}
\newcommand{\cproc}{r}
\newcommand{\dproc}{s}
\newcommand{\aprocacc}{{\pap{\aproc}{\prime}}}
\newcommand{\bprocacc}{{\pap{\bproc}{\prime}}}
\newcommand{\aproci}{\iap{\aproc}}
\newcommand{\bproci}{\iap{\bproc}}
\newcommand{\aprocacci}{\pbap{\aproc}{\prime}}

\newcommand{\univproc}{\textrm{U}}
\newcommand{\aunivproc}{{u}}
\newcommand{\aunivprocacc}{{\aunivproc'}}

\newcommand{\deadlock}{0}

\newcommand{\sjoin}{{\&}} 
\newcommand{\join}[2]{{#1}\hspace*{1.25pt}\sjoin\hspace*{1.25pt}{#2}}

\newcommand{\splus}{{+}}
\newcommand{\plus}{\mathrel{\splus}}

\newcommand{\actpref}[2]{{#1}\hspace*{0.5pt}{.}\hspace*{0.5pt}{#2}}

\newcommand{\sjoindot}{{\bap{\&}{\hspace*{-0.25pt}\bullet}}}
\newcommand{\joindot}[2]{{#1}\hspace*{1.25pt}\sjoindot\hspace*{1.75pt}{#2}}

\newcommand{\aparambisim}{\mathcal{B}}
\newcommand{\sabisim}{B}
\newcommand{\abisim}{\mathrel{\sabisim}}
\newcommand{\sabisimi}{\bap{\sabisim}}
\newcommand{\abisimi}[1]{\mathrel{\sabisimi{#1}}}
\newcommand{\aparamsim}{\mathcal{S}}
\newcommand{\sasim}{S}
\newcommand{\asim}{\mathrel{\sasim}}
\newcommand{\sasimi}{\bap{\sasim}}
\newcommand{\asimi}[1]{\mathrel{\sasimi{#1}}}

\newcommand{\sconvrel}[1]{\pap{#1}{\smallsmile}}
\newcommand{\convrel}[1]{\mathrel{\sconvrel{#1}}}

\newcommand{\slt}[1]{{\xrightarrow{#1}}}
\newcommand{\sltzero}{{\xrightarrow{}}}
\newcommand{\lt}[1]{\mathrel{\slt{#1}}}
\newcommand{\permitslt}[1]{\hspace*{1pt}\slt{#1}}
\newcommand{\snotpermitslt}[1]{\overset{#1}{\nrightarrow}}
\newcommand{\notpermitslt}[1]{\hspace*{1pt}\snotpermitslt{#1}}
\newcommand{\slti}[2]{{\xrightarrow{#1}_{#2}}}
\newcommand{\sltzeroi}[1]{\slti{}{#1}}
\newcommand{\lti}[2]{\mathrel{\slti{#1}{#2}}}

\newcommand{\sproclt}{\slt}
\newcommand{\sprocltzero}{\slt{}}
\newcommand{\proclt}[1]{\mathrel{\sproclt{#1}}}

\newcommand{\senvlt}[1]{{\xRightarrow{#1}}}
\newcommand{\senvltzero}{\senvlt{}}
\newcommand{\envlt}[1]{\mathrel{\senvlt{#1}}}

\newcommand{\ssimby}{\le}
\newcommand{\simby}{\mathrel{\ssimby}}
\newcommand{\scansimulate}{\ge}
\newcommand{\cansimulate}{\mathrel{\scansimulate}}
\newcommand{\ssimbyequiv}{{\le}\hspace*{-0.5pt}{\ge}}
\newcommand{\simbyequiv}{\mathrel{\ssimbyequiv}}
\newcommand{\sparamsimby}[1]{{\iap{\ssimby}{#1}}}
\newcommand{\paramsimby}[1]{\mathrel{\sparamsimby{#1}}}
\newcommand{\scanparamsim}[1]{{\iap{\scansimulate}{#1}}}
\newcommand{\canparamsim}[1]{\mathrel{\scanparamsim{#1}}}
\newcommand{\sparamsimequiv}[1]{{\iap{({\ssimby}{\ge})}{#1}}}
\newcommand{\paramsimequiv}[1]{\mathrel{\sparamsimequiv{#1}}}
\newcommand{\sjiparamsimby}[1]{{\bap{\ssimby}{\sjoin{#1}}}}
\newcommand{\jiparamsimby}[1]{\mathrel{\sjiparamsimby{#1}}}
\newcommand{\scanjiparamsim}[1]{{\bap{\scansimulate}{\sjoin{#1}}}}
\newcommand{\canjiparamsim}[1]{\mathrel{\scanjiparamsim{#1}}}
\newcommand{\sjiparamsimequiv}[1]{({\ssimby}\hspace*{-0.5pt}{\ge})_{\sjoin{#1}}}
\newcommand{\jiparamsimequiv}[1]{\mathrel{\sjiparamsimequiv{#1}}}
\newcommand{\siso}{{\simeq}}
\newcommand{\iso}{\mathrel{\siso}}

\newcommand{\sbisim}{\sim}
\newcommand{\bisim}{\mathrel{\sim}}
  \newcommand{\snotbisim}{\nsim}
  \newcommand{\notbisim}{\mathrel{\snotbisim}}

\newcommand{\sparambisim}[1]{{\bap{\sbisim}{#1}}}
\newcommand{\parambisim}[1]{\mathrel{\sparambisim{#1}}}
  \newcommand{\snotparambisim}[1]{{\bap{\snotbisim}{#1}}}
  \newcommand{\notparambisim}[1]{\mathrel{\snotparambisim{#1}}}
\newcommand{\sjiparambisim}[1]{{\sbisim_{\sjoin{#1}}}}  
\newcommand{\jiparambisim}[1]{\mathrel{\sjiparambisim{#1}}}
  \newcommand{\snotjiparambisim}[1]{{\bap{\snotbisim}{\sjoin{#1}}}}
  \newcommand{\notjiparambisim}[1]{\mathrel{\snotjiparambisim{#1}}}

\newcommand{\punc}[1]{\ensuremath{\hspace*{1.5pt}{#1}}}

\newcommand{\sdiscrpo}{{\sqsubseteq}}
\newcommand{\discrpo}{\mathrel{\sdiscrpo}}

%% file: abstract.tex
\begin{abstract}
  Departing from Larsen's concept of parameterized bisimilarity of processes with respect to interaction with environments,
    we start an exploration of its natural weakening: % that is defined as 
      bisimilarity of unrestricted join interactions with environments. % processes.
    % We start to investigate, for Larsen's parameterized bisimilarity of processes,
    %   the natural weakening that is defined as bisimilarity of unrestricted join interactions with environment processes.
  Parameterized bisimilarity %$\sparambisim{\aenv}$ 
                             %         of processes
    relates processes $\aproc$ and $\bproc$ with respect to an environment $\aenv$
        if $\aproc$ and $\bproc$ behave bi\nb-similarly while joining---respectively the same---transitions from~$\aenv$.       
  The weakened variant relates processes $\aproc$ and $\bproc$ with respect to environment $\aenv$
    if the join-interaction processes $\join{p}{e}$ and $\join{q}{e}$ of $\aproc$ and $\bproc$ with $\aenv$ are~bisimilar.
    % if the processes $\join{p}{e}$ and $\join{q}{e}$
    % that result from join interactions of $\aproc$ and $\bproc$ with $\aenv$ are~bisimilar.
  (Hereby join interactions $\join{r}{f}$ facilitate a step with label $a$ to $\join{r'}{f'}$     
    if and only if $r$ and $f$ permit $a$\nb-steps to $r'$ and $f'$,~respectively.)

  \smallskip  
  Join-in\-ter\-ac\-tion pa\-ram\-e\-ter\-ized (\jiparameterized) bi\-sim\-i\-larity %$\sjoinparambisim{\aenv}$ 
    coincides with parameterized bi\-sim\-i\-larity %$\sjoinparambisim{\aenv}$ %parameterized bisimilarity 
    for deterministic environments, but that it is a coarser equivalence in~general.     
  We explain how Larsen's concept can be recovered from \jiparameterized\ bisimilarity by `determinizing' interactions. %the environment.
  %Furthermore
  We show that by adaptation to simulatability (simulation \preorder)
    the same concept arises:
      parameterized simulatability coincides with \jiparameterized\ simulatability.            
  For the discrimination \preorder\ of (ji-)pa\-ram\-e\-ter\-ized\ simulatability on environments 
    we obtain the same result as Larsen did for parameterized bisimilarity.
  Also, we give a modal-logic characterization of (ji-)pa\-ram\-e\-ter\-ized simulatability. 
  Finally we gather open problems, and provide an outlook on our current~related~work.

\end{abstract}

%% file: intro.tex
With the motivation of developing flexible formal methods 
  for proving correctness of software programs incrementally,
    by showing compositional correctness under the formation of contexts, 
    % showing software to be compositionally correct under the formation of contexts,
  Larsen in \cite{lars:1986,lars:1987} introduced parameterized bisimilarity of processes as a helpful concept.
It turned out, more recently, to be useful in an area with a similar motivation:
  contextual behavioural metrics (see work of Dal~Lago and Murgia \cite{lago:murg:2023,lago:murg:2023:arxiv}),  
    which measure differences between programs as distances by means of pseudo-metrics.
      This is because parameterized bisimilarity provides natural examples of contextual~behavioural~metrics.  
  
\smallskip  
The idea underlying parameterized bisimilarity is that the behaviors of two processes are compared with respect to a third process
  that represents a common environment, in which both processes are placed, and with which both can interact separately.
The environment is able to `consume' a transition from a process by performing a transition with the same action label, 
  after which both the process and the environment move to the target state of the interaction transition on their side, respectively.
Such consumption interactions are intended to continue as long as possible. 
Yet in case that an environment state permits no transition with the same label as the current process state
  (this is the case, for example, if the environment or the considered process is in a deadlock state),
    the consumption process stops. 
Given this setup, processes $\aproc$ and $\bproc$ are called bisimilar with respect to an environment process~$\aenv$
  if $\aproc$ and $\bproc$ behave in a bisimilar way (fulfilling forth and back conditions as typical for bisimulations)
    for any pair of runs of \emph{synchronous} consumption interactions of the environment with the two processes
      in which the environment takes the \emph{same} transitions on its side. 
    
% Hereby interaction takes place in the form of a `join' operation, 
%   which means that only transitions with the same label from a process and the environment can be `joined' to interact
%     and produce a resulting transition with again that same label.
% If the environment reaches a deadlock state, in which no further transition is possible, 
%   then no further comparison of the reached states $\aprocacc$ and $\bprocacc$ has to~take~place.
  
It is distinctive for Larsen's concept of parameterized bisimilarity that 
  the forth and back conditions of two processes $\aproc$ and $\bproc$, and subsequently, of states reached via transitions from $\aproc$ and $\bproc$,
    have to be verified \emph{separately} for every run of the environment
      but while interacting \emph{synchronously} with both processes.
      Indeed, comparisons of possible further interactions have to be carried out in synchronicity of the interactions,
        as long as the environment can interact with either of the processes.
% If from processes $\aproc$ and $\bproc$ that interact with an environment~$\aenv$
%    by successful comparisons derivative processes $\aprocacc$ and $\bprocacc$ as well as derivative environment $\aenvacc$
%      are reached, for which $\aenvacc$ permits, say, an \transitionact{\aact}
%        that can be joined, for example, 
%          only with an \transitionact{\aact} from $\aprocacc$, but not from $\bprocacc$ (in case $\bprocacc$ does not permit \transitionsact{\aact})
%            then $\aproc$ and $\bproc$ are not bisimilar with respect to~$\aenv$.
%\changed{%
Thereby a mismatch is detected in the following situation:
  Suppose that 
   by successful comparisons in a synchronous run
     derivative processes $\aprocacc$ and $\bprocacc$ as well as derivative environment $\aenvacc$
     are reached.
      Suppose further that $\aenvacc$ permits, say, an \transitionact{\aact}
       that can be joined %, for example, 
         only with an \transitionact{\aact} from $\aprocacc$, but not from $\bprocacc$ (in case $\bprocacc$ does not permit \transitionsact{\aact}).
           Then it has been determined that $\aproc$ and $\bproc$ are not bisimilar~with~respect~to~$\aenv$.%}

Parameterized bisimilarity thus compares the behavior of two processes with respect to \emph{controlled} and \emph{synchronous} interactions with an environment process.
For determining whether two processes $\aproc$ and $\bproc$ are bisimilar with respect to an environment process~$\aenv$
  it is necessary to observe the consumption interactions of $\aproc$ with $\aenv$ and of $\bproc$ with $\aenv$
    in a synchronous step-wise manner.    
It is not sufficient to be merely presented the completed processes that result from the interactions of $\aproc$ with $\aenv$, and of $\bproc$ with $\aenv$, respectively,
  and then to ask whether these results are bisimilar. 
  
There are, however, conceivable practical situations, in which one lacks sufficient control over the environment process in order to perform, or merely to analyze, 
  controlled and synchronous interactions with the considered processes.   
That is, situations in which a scientist has access only to the data of completed interactions of two processes with a given environment,
  but in which she lacks sufficient control over the environment in order to perform the two interactions synchronously in a step-by-step manner
    so that she can compare the behaviors that remain after each step.
    
Here we define, and start to investigate, the weaker concept of parameterized bisimilarity
  in which only the completed outcome processes of the possible interactions of two processes $\aproc$ and $\bproc$ with a given environment are compared
    as to whether they are bisimilar.
For this purpose we stipulate that the consumption interaction takes place in the form of a `join' operation ($\sjoin$) between each process and the environment,
  which produces transitions with the same action labels as the two interaction transitions. 
    Indeed, only transitions with the same label from a process and the environment can be `joined' to interact,
    and produce a resulting transition with again that same label.
We call the concept of bisimilarity between the join interactions of each process with the environment
  `\joininteraction\ parameterized' (\jiparameterized) bisimilarity. 
Larsen briefly mentions this concept at the end of the article \cite{lars:1987}.
  He calls it `perhaps more immediate', but excludes it from further consideration on the basis that it lacks some distinctive properties %
%         \footnote{Larsen specifically mentions the lack of an immediate modal-logical characterization, 
%                     and of a characterization of the induced discrimination \preorder\ via simulatability 
%                       analogous to his result for parameterized bisimilarity, see Theorem~\ref{thm:char:discr-preorder:parambisim}.}
      that he was able to show for parameterized bisimilarity. 
      (For more details, see the paragraph `Larsen on \jiparameterized\ bisimilarity \ldots' in Section~\ref{concl}.)
Although that is true,
  it remains the case
             that \jiparameterized\ bisimilarity has a much easier, and appealingly natural definition,
             and that it may be of practical use in cases in which parameterized bisimilarity cannot~be~used.

While these considerations may seem abstract, we got interested in studying \jiparameterized\ bisimilarity
  when we made the following concrete observations (many of which are explained here later):
\begin{itemize}[itemsep=0ex]
  \item
    parameterized bisimilarity and \jiparameterized\ bisimilarity do 
    not coincide, so we try to understand the conceptual difference 
    between them, also by means of concrete examples
        (for an overview see Theorem~\ref{thm:incl:jiorparam:bisim:sim});
  \item
    noticing that, for deterministic environments, parameterized bisimilarity and \jiparameterized\ bi\-sim\-i\-la\-rity coincide
      (see~Proposition~\ref{prop:jiparamsim:equals:parasim});   
  \item  
    recognizing that also parameterized bisimilarity can be formulated as bisimilarity of a special kind ($\sjoindot$) of join interaction 
      (see Definition~\ref{def:joindot} and Lemma~\ref{lem:joindot:join});
  \item
    recognizing that simulation \preorder\ adaptations of the two concepts of parameterized bisimilarity
      and \jiparameterized\ bisimilarity
        do in fact coincide (see Proposition~\ref{prop:jiparamsim:equals:parasim}, \ref{it:1:prop:jiparamsim:equals:parasim});
  \item
%    discovering an easy adaptation of Larsens \modallogical\ characterization of parameterized bisimilarity
%      for \jiorparameterized\ simulatability\footnotemark\addtocounter{footnote}{-1} (see Theorem~\ref{thm:log-char:paramsim:paramsimequiv} in Section~\ref{log-char});% 
    Larsens \modallogical\ characterization of parameterized bisimilarity
      for \jiorparameterized\ can be easily adapted to simulatability\footnotemark\addtocounter{footnote}{-1} (see Theorem~\ref{thm:log-char:paramsim:paramsimequiv} in Section~\ref{log-char});% 
  \item
    a natural \modallogical\ characterization also for \jiparameterized\ bisimilarity is worth studying, albeit we do not
    provide a solution to it in this work
      (see current work item \ref{CW1} in Section~\ref{concl}).
\end{itemize}    
  
In Section~\ref{prelims} we summarize Larsen's definition and main results on parameterized bisimilarity,
  and we define parameterized simulatability.%
    \footnote{We use `simulatability' instead of `similarity' for `simulation preorder' for two reasons:
                to prevent the impression that a symmetrical relation were meant,
                  and to avoid a possible confusion with `simulation equivalence' that will also appear here.}
In Section~\ref{ji-param} we define \jiparameterized\ bisimilarity and simulatability,
  and develop basic results about their relationships with parameterized bisimilarity and simulatability.
    We discover that Larsen's theorem about the discrimination \preorder\ induced by parameterized bisimilarity
      has an analogous version for the discrimination \preorder\ that is induced by \jiorparameterized\ simulatability.
Then in Section~\ref{log-char}
  we specialize Larsen's \modallogical\ characterization of parameterized bisimilarity
    to \jiorparameterized\ simulatability. 
Finally in Section~\ref{concl} 
  we give a list that summarizes our results,
    we report about the literature and our ongoing related work,
      and we sketch further ideas~and~plans.

% To mention:
% %
% \begin{itemize}
%   \item
%     Larsen \cite{lars:1987}:
%     ``Though perhaps more immediate, this parameterized version $\sjiparambisim{\aenv}$ lacks many of the properties presented in this paper.'' 
%   \item  
%     \todo{Get clear the difference of properties between $\sparambisim{\aenv}$ and $\sjiparambisim{\aenv}$.}
%   
% 
% \end{itemize}

%% file: prelims.tex
In this section we summarize definitions and results by Larsen in \cite{lars:1986,lars:1987}
  concerning parameterized bisimilarity, its induced discrimination \preorder, and a \modallogical\ characterization for it.
Additonally we define parameterized simulations, which relate to parameterized bisimulations
  in the same way as how simulations relate to bisimulations.   
We start with the basic concept of labeled~transition~system.   
  
\begin{defi}[\protect\LTSs]\label{def:LTS}
  A \emph{(simple) labeled transition system (\LTS)}  
    is a triple~$\aLTS = \tuple{\states,\actions,\sltzero}$
      that consists of a set $\states$ of \emph{states},
                       a set $\actions$ of \emph{actions},
                   and a ternary \emph{transition relation} $\sltzero \subseteq \states\times\actions\times\states$
                     that represents \labeledi{\actions} \emph{transitions} on the state set.                   
\end{defi}

For \LTSs\ we will use notation and terminology for basic properties as follows.
  For their stipulation,
  we let $\aLTS = \tuple{ \states,\actions,\sltzero }$ be an \LTS. 
  For $\tuple{\astate,\aact,\bstate}\in\sltzero$ we usually write $\astate \lt{\aact} \bstate$,
    and say that ``in state~$\astate$ there is a transition with label $\aact$ (symbolizing an action called $a$) to state~$\bstate$''.
      In this case we also say that $\bstate$ is an \emph{\derivativeact{\aact}} of $\astate$.
  For $\astate\in\states$ and $\aact\in\actions$ we write $\astate \permitslt{\aact}$  
    if there is an \transitionact{\aact} from $\astate$ in $\aLTS$,
      and $\astate \notpermitslt{\aact}$ if there is no \transitionact{\aact} from $\astate$ in $\aLTS$.
  
  We call an \LTS~$\aLTS = \tuple{ \states,\actions,\sltzero }$ \emph{deterministic} (respectively \emph{\imagefinite})
    if $ \cardinalityns{ \descsetexpns{ \astateacc }{ \astate \lt{\aact} \astateacc } } \le 1 $
       (and respectively if $\cardinalityns{ \descsetexpns{ \astateacc }{ \astate \lt{\aact} \astateacc } } < \infty $ )
       for all states $\astate\in\states$ and actions~$\aact\in\actions$,
  that is, if every state of $\aLTS$ has at most one \derivativeact{\aact} 
                                     (resp.\ has only finitely many \derivativesact{\aact}), for all $\aact\in\actions$. 
  We say that a state~$\astate\in\states$
   \emph{is deterministic} (resp.\ \emph{is \imagefinite})
     if every state of $\aLTS$ that is reachable from $\astate$ via a path of transitions 
               has at most one \derivativeact{\aact} 
       (resp.\ has only finitely many \derivativesact{\aact}), for all $\aact\in\actions$.

\smallskip

% \begin{nota}[notions for, properties of \LTSs]
%   Let $\aLTS = \tuple{ \states,\actions,\sltzero }$ be an \LTS. 
%   
%   For $\tuple{\astate,\aact,\bstate}\in\sltzero$ we usually write $\astate \lt{\aact} \bstate$,
%     and say that ``in state~$\astate$ the system can perform an action with label $\aact$ to state~$\bstate$''.
%   For $\astate\in\states$ and $\aact\in\actions$ we write $\astate \permitslt{\aact}$  
%     if there is an \transitionact{\aact} from $\astate$ in $\aLTS$,
%       and $\astate \notpermitslt{\aact}$ if there is no \transitionact{\aact} from $\astate$ in $\aLTS$.
%   
%   We call $\aLTS$ \emph{deterministic} (resp.\ \emph{\imagefinite})
%     if $ \cardinalityns{ \descsetexpns{ \astateacc }{ \astate \lt{\aact} \astateacc } } \le 1 $
%        (and resp.\ if $\cardinalityns{ \descsetexpns{ \astateacc }{ \astate \lt{\aact} \astateacc } } < \infty $ )
%        for all states $\astate\in\states$ and actions~$\aact\in\actions$,
%   that is, if every state of $\aLTS$ has at most one \derivativeact{\aact} 
%                                      (resp.\ has only finitely many \derivativesact{\aact}), for all $\aact\in\actions$. 
%   We say that a state~$\astate\in\states$
%    \emph{is deterministic} (resp.\ \emph{is \imagefinite})
%      if every state of $\aLTS$ that is reachable from $\astate$ via a path of transitions 
%                has at most one \derivativeact{\aact} 
%        (resp.\ has only finitely many \derivativesact{\aact}), for all $\aact\in\actions$. 
% \end{nota}

Following well-known intuitions, %about these concepts, 
  bi-/simulations on such simple \LTSs\ can be defined as follows.  

% Following classic intuitions about these concepts, bi-/simulations on \LTSs\ can be defined as follows.   

\begin{defi}[bisimulation\hspace*{0.5pt}/\hspace*{0.5pt}bisimilar, simulation\hspace*{0.5pt}/\hspace*{0.5pt}simulated~by]%
    \label{def:bisim:sim}
  Let $\aLTS = \tuple{ \states, \actions, \sprocltzero }$ be an \LTS.
  
  \begin{enumerate}[label={(\roman{*})},align=right,itemsep=0ex]
    \item{}\label{it:1:def:bisim:sim}
      A \emph{bisimulation $\abisim$ on $\aLTS$}
        is a \nonempty\ binary relation $\sabisim \subseteq \states \times \states$ with the following property:
        If $\astate \abisim \bstate$ for $\astate,\bstate\in\states$,
            then the following two conditions hold:%\vspace{-1ex}   
          \begin{enumerate}[align=right,leftmargin=4em,itemsep=0ex]
            \item[(forth)]\label{forth:bisim}
              $ \;
                (\forall \astateacc\in\states)
                  \bigl[\, \astate \proclt{\aact} \astateacc 
                             \;\;\Longrightarrow\;\;
                               (\exists\, \bstateacc\in\states)
                                     [\, \bstate \proclt{\aact} \bstateacc 
                                            \logand
                                          \astateacc \abisim \bstateacc \,] \bigr ]
                                          $,
            \item[(back)]\label{back:bisim}  
              $ \;
                (\forall \bstateacc\in\states)
                  \bigl[\, \bstate \proclt{\aact} \bstateacc 
                             \;\;\Longrightarrow\;\;
                               (\exists\, \astateacc\in\states)
                                 [\, \astate \proclt{\aact} \astateacc 
                                        \logand
                                      \astateacc \abisim \bstateacc ] \bigr]
                                          $.

         \end{enumerate}
      For processes $\astate,\bstate\in\states$,
        we write $\astate \bisim \bstate$ and say that \emph{$\astate$ and $\bstate$ are bisimilar}
          if there is a bisimulation $\sabisim$ on $\aLTS$ such that $\astate \abisim \bstate$.
    \item{}\label{it:2:def:bisim:sim}  
      A \emph{simulation $\abisim$ on $\aLTS$}
        is a \nonempty\ binary relation $\sasim \subseteq \states \times \states$ with the following property: 
        If $\astate \abisim \bstate$ for $\astate,\bstate\in\states$,
          then the condition (forth) in \ref{it:1:def:bisim:sim} holds for $\sabisim \defdby \sasim$
          (but \underline{\smash{not}} necessarily the condition (back)).
      For processes $\astate,\bstate\in\states$,
        we write $\astate \simby \bstate$ and say that \emph{$\astate$ can be simulated by $\bstate$},
        and we write $\bstate \cansimulate \astate$ and say that \emph{$\bstate$ can simulate $\astate$},
          if there is a simulation $\sasim$ on $\aLTS$ such that $\astate \asim \bstate$.
  \end{enumerate}
\end{defi}

%It should be noted that
Rather than defining simulations as weakened versions of bisimulations as above,
  bisimulations can also be defined from simulations, as follows.
    A relation~$\abisim$ on an \LTS~$\aLTS$ is a bisimulation 
      if and only if
    both $\abisim$ as its converse relation $\convrel{\abisim} \defdby \descsetexpns{ \pair{\bstate}{\astate} }{ \astate \abisim \bstate }$
      are simulations on $\aLTS$.

For modeling processes whose behavior is studied 
  according to how they interact with environments,
    both processes and environments are formalized as \LTSs.
However, in order to indicate their intended roles for occurring \LTSs,
  we distinguish in notation, name, and in how they are referenced
    between \emph{process \LTSs} $\aprocLTS = \tuple{ \procs, \actions, \sprocltzero }$,
      whose states we call \emph{processes},
        and \emph{environment \LTSs} $\aenvLTS = \tuple{ \envs, \actions, \senvltzero }$,\vspace{-1.5pt}
          whose states we call \emph{environments}.
We follow Larsen \cite{lars:1986,lars:1987} in this terminology and in most of the notation.
Based on this distinction, Larsen defines parameterized bisimulation and bisimilarity~as~follows.

\begin{defi}[parameterized bisimulation (Larsen \cite{lars:1986,lars:1987})]%
    \label{def:parambisim}
  Let $\aprocLTS = \tuple{ \procs, \actions, \sprocltzero }$ be a process \LTS, and let $\aenvLTS = \tuple{ \envs, \actions, \senvltzero }$ be an environment \LTS. 
  An \emph{\parameterizedby{\aenvLTS} bisimulation $\aparambisim$ on $\aprocLTS$}\vspace{-2.5pt}
    is an \indexedby{\envs} family $\aparambisim = \familyns{\sabisimi{\benv}}{\benv\in\envs}$ of 
      \nonempty\ binary relations $\sabisimi{\benv} \subseteq \procs \times \procs$ % for $\benv\in\envs$
      such that the following holds:\footnote{Note the occurrence of $\aenvacc$ (instead of $\aenv$) in $\sabisimi{\aenvacc}$ in both of the conditions (back) and (forth).}
      If $\aproc \abisimi{\aenv} \bproc$ for %$\aproc,\bproc\in\procs$ and 
                                             $\aenv\in\envs$, % (thus $\aproc,\bproc\in\procs$),
        then if $\aenv \envlt{\aact} \aenvacc$ for $\aact\in\actions$ %and $\aenvacc\in\envs$ 
          the following conditions~hold:\vspace{0ex}   
      \begin{enumerate}[align=left,leftmargin=2em,itemsep=0ex]
        \item[(forth)]\label{forth:parabisim}
          $ (\forall \aprocacc\in\procs)
              \bigl[\, \aproc \proclt{\aact} \aprocacc 
                         \;\;\Longrightarrow\;\;
                           (\exists\, \bprocacc\in\procs)
                                 [\, \bproc \proclt{\aact} \bprocacc 
                                        \logand
                                      \aprocacc \abisimi{\aenvacc} \bprocacc \,] \bigr ]
                                      $, 
        \item[(back)]\label{back:parabisim}  
          $ (\forall \bprocacc\in\procs)
              \bigl[\, \bproc \proclt{\aact} \bprocacc 
                         \;\;\Longrightarrow\;\;
                           (\exists\, \aprocacc\in\procs)
                             [\, \aproc \proclt{\aact} \aprocacc 
                                    \logand
                                  \aprocacc \abisimi{\aenvacc} \bprocacc ] \bigr]
                                      $.
      \end{enumerate}
  For processes $\aproc,\bproc\in\procs$, and environments $\aenv\in\envs$
    we write $\aproc \parambisim{\aenv} \bproc$ and say that \emph{$\aproc$ and $\bproc$ are bisimilar with respect to $\aenv$}
      if there is an \parameterizedby{\aenvLTS} bisimulation $\aparambisim = \familyns{\sabisimi{\benv}}{\benv\in\envs}$ such that $\aproc \abisimi{\aenv} \bproc$.
\end{defi}

While simulation plays a crucial role in Larsen's main theorem on parameterized bisimulation, see Theorem~\ref{thm:char:discr-preorder:parambisim} below, 
  it is surprising that he did not also define the simulation version of this concept with only the forth condition from its progression conditions.
For the reason   
  that it can be linked directly to the simulation version of the concept of `\jiparameterized\ bisimulation'
    that we will introduce in Section~\ref{ji-param} (see Definition~\ref{def:jiparamsimby:jiparamsimequiv:jiparambisimby}),
      % (see Definition~\ref{def:jiparamsimby:jiparamsimequiv:jiparambisimby} in Section~\ref{ji-param}),
    we also define `parameterized simulation' here.

\begin{defi}[parameterized simulation]\label{def:paramsim}
  Let $\aprocLTS = \tuple{ \procs, \actions, \sprocltzero }$ be a process \LTS,\vspace{-2.5pt}
    and let $\aenvLTS = \tuple{ \envs, \actions, \senvltzero }$ be an environment \LTS.  
  An \emph{\parameterizedby{\aenvLTS} simulation~$\aparamsim$ on $\aprocLTS$}
    is an \indexedby{\envs} family $\aparamsim = \familyns{\sasimi{\benv}}{\benv\in\envs}$ of \nonempty\ binary relations $\asimi{\benv} \,\subseteq\, \procs \times \procs$ 
        % for $\benv\in\envs$
      such that the following holds:
      If $\aproc \asimi{\aenv} \bproc$ %for $\aproc,\bproc\in\procs$ and 
                                       holds for $\aenv\in\envs$,
        then for all $\aact\in\actions$, if $\aenv \envlt{\aact} \aenvacc$
          the condition (forth) in Def.~\ref{def:parambisim} for $\sabisimi{\aenvacc} \defdby \asimi{\aenvacc}$ holds:%\vspace{-1ex}
      \begin{enumerate}[align=left,leftmargin=2em,itemsep=0ex]
        \item[(forth)]
          $ (\forall \aprocacc\in\procs)
              \bigl[\, \aproc \proclt{\aact} \aprocacc 
                         \;\;\Longrightarrow\;\;
                           (\exists\, \bprocacc\in\procs)
                                 [\, \bproc \proclt{\aact} \bprocacc 
                                        \logand
                                      \aprocacc \asimi{\aenvacc} \bprocacc \,] \bigr ]
                                      $.
      \end{enumerate}
  For processes $\aproc,\bproc\in\procs$, and environments $\aenv\in\envs$
    we write $\aproc \paramsimby{\aenv} \bproc$ and say that \emph{$\aproc$ can be simulated by~$\bproc$ with respect to $\aenv$},
      and $\bproc \canparamsim{\aenv} \aproc$ and say that \emph{$\bproc$ can simulate $\aproc$ with respect to $\aenv$},
        if there is an \parameterizedby{\aenvLTS} simulation $\aparamsim = \familyns{\sasimi{\benv}}{\benv\in\envs}$ such that $\aproc \asimi{\aenv} \bproc$.
\end{defi}

\smallskip

Also parameterized bisimulations can be defined from parameterized simulations:
  for process \LTS~$\aprocLTS$, and environment \LTS~$\aenvLTS$,
    $\aparamsim = \familyns{\sasimi{\benv}}{\benv\in\envs}$ is an \parameterizedby{\aenvLTS} bisimulation on $\aprocLTS$
      if and only if
    $\familyns{\sasimi{\benv}}{\benv\in\envs}$, and
     the family $\familyns{\convrel{\sasimi{\benv}}}{\benv\in\envs}$ of converse relations of $\sasimi{\benv}$
       are \parameterizedby{\aenvLTS} simulations.

\smallskip
Larsen's main result on parameterized bisimilarity 
  concerns the \emph{discrimination \preorder~$\sdiscrpo$} that orders environments
    according to their power of discriminating between processes.
          % and which is defined by
          %   $\,\aenv \discrpo \benv \;\funin \,\Leftrightarrow\; \snotparambisim{\aenv} \subseteq \snotparambisim{\benv}
          %                                    \; ( \Leftrightarrow\; \sparambisim{\benv} \subseteq \sparambisim{\aenv} ) \,$,
          %   for all environments $\aenv$ and $\benv$.     
    It is defined, for a given process \LTS~$\aprocLTS = \tuple{\procs,\actions,\sprocltzero}$ 
                   and a given environment \LTS~$\aenvLTS = \tuple{\envs,\actions,\senvltzero}$
                   (so that\vspace{-2pt}, for all $\aenv\in\envs$,
                    the relations $\sparambisim{\aenv}$ are then fixed as subsets of $\procs\times\procs$),
    for all $\aenv,\benv\in\envs$ by:
    \begin{equation}\label{eq:def:discrpo}
      \aenv \discrpo \benv \;\funin \,\Longleftrightarrow\; \snotparambisim{\aenv} \subseteq \snotparambisim{\benv}
                             \qquad ( \Longleftrightarrow\; \sparambisim{\benv} \subseteq \sparambisim{\aenv} ) \punc{.}
    \end{equation}      
Now Larsen's result characterizes $\sdiscrpo$ as coinciding with the simulation \preorder~$\ssimby$ on environments.
For the `completeness' direction of this characterization to hold (``$\Leftarrow$'' in \eqref{eq:thm:char:discr-preorder:parambisim}),
  it is, however, necessary to assume that the underlying process \LTS\ is, as Larsen formulates it, `sufficiently rich' structurally.
    The weak natural assumption that he uses for the purpose of guaranteeing sufficient structural richness of any considered process \LTS\
      is that its set of processes is closed under action prefixing and finite~summation
        (see Definition~\ref{def:LTS:closure} in Section~\ref{ji-param}).
     
\begin{thm}[Larsen \cite{lars:1986,lars:1987}]\label{thm:char:discr-preorder:parambisim}
  The following logical equivalence holds,
    provided that the underlying process \LTS\ is closed under action prefixing and finite summation,
      for all \imagefinite\ environments~$\aenv,\,\benv\,$:
  \begin{align}
    \aenv \simby \benv
       \;\; & \:\Longleftrightarrow\;\;\;\:
     \aenv \discrpo \benv
      \quad
            %  ( {} \Longleftrightarrow\;\;
            %        \snotparambisim{\aenv}
            %          \;\,\subseteq\;\,
            %        \snotparambisim{\benv} 
             ( {} \Longleftrightarrow\;\;
                    \sparambisim{\benv}
                      \;\subseteq\;
                    \sparambisim{\aenv} 
              ) 
      \punc{.}
        \label{eq:thm:char:discr-preorder:parambisim}
  \end{align}
  The implication ``$\Rightarrow$'' holds for all (thus also for not necessarily \imagefinite) environments $\aenv$ and $\benv$. %,
    %and all underlying process \LTSs.
    %
\end{thm}

Specifically for the direction ``$\Leftarrow$'' in \eqref{eq:thm:char:discr-preorder:parambisim}
  Larsen provides an impressive, technical proof,
    which he found, as he writes, only after an intensive search that took several months.

% %----------
% \subsection{Logical characterization of (bi-)similarity and of parameterized bisimilarity}
% %----------

\medskip
We now turn to \modallogical\ characterizations of the relations of being able to be simulated by~$\ssimby$, of bisimilarity~$\sbisim$, 
  and of parameterized bisimilarity~$\sparambisim{\aenv}$.
For expressing properties of \LTSs\  
  such as the existence of a transition with label~$\aact$ from a given state such that at the target state property $\aformi{0}$ holds,
    modal formulas should include a diamond modality $\diamondact{\aact}$ to build formulas like $\diamondact{\aact} \aformi{0}$. 
The set $\forms$ of simple modal formulas
   (and the set $\posforms$ of positive formulas)
  are now defined with these diamond modalities and basic propositional connectives (resp.\ such connectives except negation) as constructors.

\begin{defi}[modal formulas]\label{def:forms}
  For given sets $\actions$ of actions,
    we define the following classes of formulas:
      $\posformsof{\actions}$ of positive formulas,
        and
      $\formsof{\actions}$ of (simple modal logic) formulas,
        via the following grammars:
  \begin{alignat}{3}
    \posformsof{\actions}
      & \qquad\quad & 
        \aform 
          & \;\BNFdefdby\;
            \True 
              \BNFor \aform \logand \aform
              \BNFor \diamondact{\aact} \aform
              & \text{(where $\aact\in\actions$),}
      \label{eq:1:def:forms}          
    \displaybreak[0]\\
    \formsof{\actions}
      & & 
        \aform 
          & \;\BNFdefdby\;
            \True 
              \BNFor \lognot{\aform}
              \BNFor \aform \logand \aform
              \BNFor \diamondact{\aact} \aform
              \qquad
              & \text{(where $\aact\in\actions$).}
      \label{eq:2:def:forms} 
  \end{alignat}
  As above, we usually will keep the underlying set $\actions$ of actions implicit,
  and write
  $\posforms$ and $\forms$
  for
  $\posformsof{\actions}$ and $\formsof{\actions}$, 
  respectively.
\end{defi}

\begin{defi}[satisfaction relation, sets of satisfied formulas]\label{def:satisfies}
  Let $\aprocLTS = \tuple{\procs,\actions,\sprocltzero}$ be a process \LTS. 
  The \emph{satisfaction relation} $\ssatisfies \subseteq \procs\times\forms$ \emph{on $\aprocLTS$}
    is defined by the following clauses:
  \begin{align*}
    \aproc \satisfies \True
      \;\; \funin & \Longleftrightarrow \;\;
        \aproc\in\procs \punc{,}
    &
    \aproc \satisfies \aformi{1} \logand \aformi{2}
      \;\; \funin & \Longleftrightarrow \;\;
        \aproc \satisfies \aformi{1}
          \text{ and }
        \aproc \satisfies \aformi{2} \punc{,}
    \\
    \aproc \satisfies \lognot{\aformi{0}}
      \;\; \funin & \Longleftrightarrow \;\;
        \aproc \notsatisfies \aformi{0} \punc{,} 
    &  
    \aproc \satisfies \diamondact{\aact} \aformi{0}
      \;\; \funin & \Longleftrightarrow \;\;
        \existsstzero{\aprocacc\in\procs\,}
          (\, \aproc \proclt{\aact} \aprocacc
                \text{ and }
              \aprocacc \satisfies \aformi{0} \,)  \punc{.}
  \end{align*}
  by induction on the structure of formulas in $\forms\!$.
  For all processes $\aproc\in\procs$ we define by:
  \begin{align*}
    \posformsof{\aproc}
      & \,\defdby\,
          \descsetexp{ \aform\in\posforms }{ \aproc \satisfies \aform } \punc{,}
    &
    \formsof{\aproc}
      & \,\defdby\,
          \descsetexp{ \aform\in\forms }{ \aproc \satisfies \aform } \punc{,}
  \end{align*}
  the set $\posformsof{\aproc}$ of positive formulas in $\posforms$ that are satisfied in $\aproc$, and respectively,
  the set $\formsof{\aproc}$ of formulas in $\forms$ that are satisfied in $\aproc$.
  % (In doing so we assume that a set $\actions$ of actions is clear from the context,
  %   with respect to which the classes $\posforms$ and $\forms$ are viewed as the sets $\posformsof{\actions}$ and $\formsof{\actions}$, respectively.)
     %
\end{defi}

The classical characterization result via \modallogical\ formulas
  of the relations bisimilarity~$\sbisim$, and `being able to be simulated by'~$\ssimby$
   is the following well-known theorem by Hennessy and Milner.

\begin{thm}[Hennessy, Milner \cite{henn:miln:1985}]\label{thm:log-char:sim:bisim}
  For all \imagefinite\ processes $\aproc$ and $\bproc$ the following
    statements~hold:
  \begin{align}
    \begin{split}
      \aproc \simby \bproc 
            %   & \;\;\Longleftrightarrow\;\;  
            %     \forallstzero{\aform\in\posforms}\,
            %       \bigl[\,
            %         \aproc \satisfies \aform
            %           \:\Rightarrow\:
            %         \bproc \satisfies \aform
            %       \,\bigr] \punc{,}
            %   \\  
        & \;\;\Longleftrightarrow\;\; 
          \posformsof{\aproc}
            \subseteq
          \posformsof{\bproc} \punc{,}
    \end{split}
    \label{eq:sim:thm:log-char:sim:bisim}  
    \\
    \begin{split}
      \aproc \bisim \bproc 
        & \;\;\Longleftrightarrow\;\;  
          %       \forallstzero{\aform\in\posforms}\,
          %         \bigl[\,
          %           \aproc \satisfies \aform
          %             \:\Leftrightarrow\:
          %           \bproc \satisfies \aform
          %         \,\bigr] \punc{,}
          %     \\
          % & \;\;\Longleftrightarrow\;\; 
          \formsof{\aproc}
            =
          \formsof{\bproc} \punc{.}
    \end{split}
    \label{eq:bisim:thm:log-char:sim:bisim} 
  \end{align}
  The implications ``$\Rightarrow$'' hold for all (thus also for not necessarily \imagefinite) processes $\aproc$ and $\bproc$.
\end{thm}

For a \modallogical\ characterization of parameterized bisimilarity,
  the concept of negation closure of positive formulas will be needed. 
    By departing slightly from Larsen's exposition in \cite{lars:1986,lars:1987} 
      we define it via a projection of general formulas to positive formulas.

\begin{defi}[positive-formula projection, negation closure]
  The \emph{pos\-i\-tive-for\-mu\-la projection} is the function 
    $\posform{\cdot} \funin \forms \longrightarrow \posforms$
      that maps formulas $\aform\in\forms$ to positive formulas $\posform{\aform}\in\posforms$,
  and that is defined by induction on the structure of $\aform$ via the following clauses,
    for all formulas $\aformi{0}, \aformi{1}, \aformi{2}\in\forms\,$:
  \begin{align*}
    \posform{\top}
      & \defdby 
          \top \,\punc{,}
    &
    \posform{\lognot{\aformi{0}}}
      & \defdby
          \posform{\aformi{0}} \,\punc{,}       
    &
    \posform{ \aformi{1} \logand \aformi{2} }
      & \defdby
          \posform{ \aformi{1} } \logand \posform{ \aformi{2} } \,\punc{,}
    &
    \posform{ \diamondact{\aact} \aformi{0} } 
      & \defdby
          \diamondact{\aact} \posform{ \aformi{0} } \,\punc{.}          
  \end{align*}
  
  For every positive formula~$\aform\in\posforms$,
    we define by
      $\negclosure{\aform} \defdby \descsetexpns{\bform\in\forms}
                                                {\posform{\bform} = \aform }$
    the \emph{negation closure of~$\aform$ in $\forms$}.
  For subclasses $\aclassforms \subseteq \posforms$ of positive formulas,
    we define the \emph{negation closure of $\aclassforms$ in $\forms$}
      by
      $\negclosure{\aclassforms} 
         \defdby
           \bigcup \descsetexpns{ \negclosure{\aform} }{ \aform\in\aclassforms }
         = \descsetexp{ \bform\in\forms }
                      { \posform{\bform} \in \aclassforms }$.
\end{defi}

Larsen presents \cite{lars:1986,lars:1987} the following modal characterization theorem of parameterized bisimilarity,
  which he attributes to Colin Stirling. The characterization restricts consideration for possible discriminating formulas 
    to those in the negation-closure of positive formulas that are satisfied by the environment.

\begin{thm}[Stirling and Larsen, \cite{lars:1986,lars:1987}]\label{thm:log-char:parambisim}
  For all \imagefinite\ processes $\aproc$, $\bproc$, and environments $\aenv\,$:$\,$\footnotemark % it holds:
    \footnotetext{The condition of being \imagefinite\ can be dropped for the environments~$\aenv$.
                  This can be verified by means of a careful analysis of the proof in \cite{lars:1986,lars:1987}
                    for this logical characterization.}
  \begin{equation}
    \begin{aligned}[t]
      \aproc \parambisim{\aenv} \bproc 
        & \;\;\overset{(\star)}{\Longleftrightarrow}\;\; 
          \formsof{\aproc}
            \cap
          \negclosure{\posformsof{\aenv}}
            =
          \formsof{\bproc} 
            \cap
          \negclosure{\posformsof{\aenv}}  
          \\
        & \;\;\Longleftrightarrow\;\;   
          \forallstzero{\aformi{0}\in\posforms}\,
            \bigl[\,
              \aenv \satisfies \aformi{0}
                \;\:\Rightarrow\:\;
                  \forallstzero{\aform\in\!\!\negclosure{\aformi{0}}}\,
                    \bigl(\,
                      \aproc \satisfies \aform
                        \;\:\Leftrightarrow\;\:
                      \bproc \satisfies \aform
                    \,\bigr)
              \,\bigr] \punc{.}
    \end{aligned}
  \end{equation}
  The implication ``$\Rightarrow$'' in $(\star)$ holds for all (thus also for not necessarily \imagefinite) $\aproc$, $\bproc$, and $\aenv$.
\end{thm}

%% file: ji-param.tex
In this section we first define 
  the weaker versions of parameterized bisimilarity and simulatability (simulation \preorder) 
    that are based on a definition of `\joininteraction' of \LTSs\ (Definition~\ref{def:joindot}):
      \jiparameterized\ simulatability and bisimilarity
        (Definition~\ref{def:jiparamsimby:jiparamsimequiv:jiparambisimby}).
Then we investigate the basic relationship between the new concepts
  and parameterized bisimilarity and simulatability (Theorem~\ref{thm:incl:jiorparam:bisim:sim}),
    and explain that also parameterized bisimilarity and simulatability can be viewed
      as bisimilarity and simulatability, resp., %, respectively, 
        with respect to a special kind of join operation (Lemma~\ref{lem:joindot:join}).
Finally we present a theorem (Theorem~\ref{thm:char:discr-preorder:jiparamsim:jiparamsimequiv})
  that characterizes the discrimination \preorder\ of \jiorparameterized\ simulatability
    in analogy with Larsen's characterization of the discrimination \preorder\
      of parameterized bisimilarity, see Theorem~\ref{thm:char:discr-preorder:parambisim}.

%----------
\subsection*{Definitions of \protect\jiparameterized\ simulatablity and bisimilarity}
%----------
 
In order to prepare for the definition of \jiparameterized\ bisimilarity
  we define `\joininteraction\ LTSs' by using an operation of processes that Larsen calls `join' \cite[p.$\,$43,44]{lars:1986}.
Later we also need the subsequent stipulation of when a single \LTS\ is closed under the `join' operation. 

\begin{defi}[\joininteraction\ of \protect\LTSs]\label{def:join}
  Let $\aLTSi{1} = \tuple{\procsi{1},\actions,\sltzeroi{1}}$
    and  $\aLTSi{2} = \tuple{\procsi{2},\actions,\sltzeroi{2}}$
      two \LTSs. 
  By the \emph{\joininteraction\ of $\aLTSi{1}$ and $\aLTSi{2}$}
    we mean the \LTS~$\join{\aLTSi{1}}{\aLTSi{2}} = \tuple{ %\procs,
                                                            \join{\procsi{1}}{\procsi{2}}, \actions, \sltzero }$
      where $ \sltzero \subseteq %\procs
                                 (\join{\procsi{1}}{\procsi{2}})\times\actions\times%\procs
                                                                                    (\join{\procsi{1}}{\procsi{2}}) $
      with $%\procs
            \join{\procsi{1}}{\procsi{2}} \defdby \descsetexpns{ \join{\aproci{1}}{\aproci{2}} }{ \aproci{1}\in\procsi{1},\, \aproci{2}\in\procsi{2} }$ 
      is defined via 
        the %following 
            transition system rule: 
      \begin{equation*}
        \AxiomC{$ \aproci{1} \lti{\aact}{1} \aprocacci{1} $}
        \AxiomC{$ \aproci{2} \lti{\aact}{2} \aprocacci{2} $}
        \BinaryInfC{$ \join{\aproci{1}}{\aproci{2}}  \lt{\aact} \join{\aprocacci{1}}{\aprocacci{2}} $}
        \DisplayProof
      \end{equation*}
      %
      %Hereby
      The symbol ``$\sjoin$'' in processes $\join{\aproci{1}}{\aproci{2}}$ of %the \joininteraction\ \LTS\ 
                                                                      $\join{\aLTSi{1}}{\aLTSi{2}}$ 
        is to be understood as a term constructor that from any two processes $\bproci{1}$ in $\procsi{1}$ and $\bproci{2}$ in $\procsi{2}$
          constructs a formal \joininteraction\ process $\join{\aproci{1}}{\aproci{2}}$ in $\join{\procsi{1}}{\procsi{2}}$. 
\end{defi}

\begin{defi}[closure of an \protect\LTS\ under action prefixing, sum, and join]\label{def:LTS:closure}
  Let $\aLTS = \tuple{\states,\actions,\sltzero}$ be a labeled transition system.
  We say that $\aLTS$ \emph{is closed under action prefixing}, resp.\ \emph{under sum}, and resp.\ \emph{under join}
    if for every states $\astate,\astatei{1},\astatei{2}\in\states$
      there exists a state $\actpref{\aact}{\astate}\in\states$ for all $\aact\in\actions$,
      resp.\ there exists a state $\astatei{1} \plus \astatei{2}\in\states$,
      and resp.\ there exists a state $\join{\astatei{1}}{\astatei{2}} \in \states$
        such that the respective transition rule below is satisfied:
  \begin{center}
    $
    \begin{aligned}
      \AxiomC{$ \phantom{\actpref{\aact}{\astate} \lt{\aact} \astate}\rule{0pt}{13.5pt} $}
      \UnaryInfC{$ \actpref{\aact}{\astate} \lt{\aact} \astate $}
      \DisplayProof
      & \qquad &
      \AxiomC{$ \astatei{i} \lt{\aact} \astateacci{i} $}
      \RightLabel{{\small (where $i\in\setexp{1,2}$)}}
      \UnaryInfC{$ \astatei{1} \plus \astatei{2}  \lt{\aact} \astateacci{i} $}
      \DisplayProof 
      & \qquad &
      \AxiomC{$ \astatei{1} \lt{\aact} \astateacci{1} $}
      \AxiomC{$ \astatei{2} \lt{\aact} \astateacci{2} $}
      \BinaryInfC{$ \join{\astatei{1}}{\astatei{2}}  \lt{\aact} \join{\astateacci{1}}{\astateacci{2}} $}
      \DisplayProof
    \end{aligned}
    $
  \end{center}
\end{defi}

\smallskip
We now proceed to defining \joininteraction\ versions 
  of parameterized simulatability, simulation equivalence, and bisimilarity
    as simulatability, simulation equivalence, and bisimilarity, respectively,
      of \joininteraction{s} between two processes and an environment.

\begin{defi}\label{def:jiparamsimby:jiparamsimequiv:jiparambisimby}
  Let $\aprocLTS = \tuple{ \procs, \actions, \sprocltzero }$ be a process \LTS, 
    and let $\aenvLTS = \tuple{ \envs, \actions, \senvltzero }$ be an environment \LTS.\vspace{-1pt}
  
  For all environments~$\aenv\in\envs$,
    we define three binary relations on $\procs\,$:
      \emph{\jiparameterized\ simulatability}~$\sjiparamsimby{\aenv}$,
      %then 
      \emph{\jiparameterized\ bisimilarity}~$\sjiparambisim{\aenv}$,
      and finally,
      \emph{\jiparameterized\ simulation equivalence}~$\sjiparamsimequiv{\aenv}$
   where
     $\sjiparamsimby{\aenv}, \sjiparambisim{\aenv}, \sjiparamsimequiv{\aenv} \,\subseteq\, \procs\times\procs$,
   are defined by the following clauses, for all processes $\aproc,\bproc\in\procs\,$: 
  \begin{align}
    \parbox[c]{\widthof{(with respect to \joininteraction\ with $\aenv$)}}
              {({\em $\aproc$ can be simulated by $\bproc$ 
               \\[-0.25ex]\phantom{(}%
                 with respect to \joininteraction\ with $\aenv$})} 
    & &
    \aproc \jiparamsimby{\aenv} \bproc
      \; \funin & \: \Longleftrightarrow \;\;
        \join{\aproc}{\aenv}
          \:\simby\:
        \join{\bproc}{\aenv} \punc{,} 
          \label{eq:jiparamsim:def:jiparamsimby:jiparamsimequiv:jiparambisimby}
        \\    
    \parbox[c]{\widthof{(with respect to \joininteraction\ with $\aenv$)}}
              {({\em $\aproc$ is bisimilar to $\bproc$
               \\[-0.25ex]\phantom{(}%
                 with respect to \joininteraction\ with $\aenv$})} 
    & &
    \aproc \jiparambisim{\aenv} \bproc
      \; \funin & \: \Longleftrightarrow \;\;
        \join{\aproc}{\aenv}
          \:\bisim\:
        \join{\bproc}{\aenv} \punc{,}
          \label{eq:jiparambisim:def:jiparamsimby:jiparamsimequiv:jiparambisimby}
        \\
    \parbox[c]{\widthof{(with respect to \joininteraction\ with $\aenv$)}}
              {({\em $\aproc$ and $\bproc$ are simulation equivalent
               \\[-0.25ex]\phantom{(}%
               with respect to \joininteraction\ with $\aenv$})} 
    & &     
    \aproc \jiparamsimequiv{\aenv} \bproc
      \; \funin & \: \Longleftrightarrow \;\;
        \aproc \jiparamsimby{\aenv} \bproc
          \;\logand\;
        \bproc \jiparamsimby{\aenv} \aproc \punc{,} 
          \label{eq:jiparamsimequiv:def:jiparamsimby:jiparamsimequiv:jiparambisimby} 
  \end{align}  
  where $\join{\aproc}{\aenv}$ and $\join{\bproc}{\aenv}$ on the right 
    in \eqref{eq:jiparamsim:def:jiparamsimby:jiparamsimequiv:jiparambisimby} and in \eqref{eq:jiparambisim:def:jiparamsimby:jiparamsimequiv:jiparambisimby}
      are processes from the \joininteraction~\LTS~$\join{\aprocLTS}{\aenvLTS}$.
  By $\scanjiparamsim{\aenv}$ we denote the converse of $\sparamsimby{\aenv}$,
    and express $\bproc \canjiparamsim{\aenv} \aproc$ verbally by saying that \emph{$\bproc$ can simulate $\aproc$ with respect to \joininteraction\ with $\aenv$}.
\end{defi}

% \begin{defi}\label{def:jiparamsimby:jiparamsimequiv:jiparambisimby}
%   For processes $\aproc$ and $\bproc$, and environments~$\aenv$,
%     we define \emph{\jiparameterized\ similarity} of $\aproc$ by $\bproc$ with respect to $\aenv$, denoted by $\aproc \jiparamsimby{\aenv} \bproc$, 
%          then \emph{\jiparameterized\ simulation equivalence} of $\aproc$ and $\bproc$ with respect to $\aenv$, denoted by $\aproc \jiparamsimequiv{\aenv} \bproc$,
%          and finally 
%               \emph{\jiparameterized\ similarity} of $\aproc$ and $\bproc$ with respect to $\aenv$, denoted by $\aproc \jiparambisim{\aenv} \bproc$,
%               by: 
%               %
%   \begin{align}
%     \aproc \jiparamsimby{\aenv} \bproc
%       \; \funin & \: \Longleftrightarrow \;\;
%         (\join{\aproc}{\aenv})
%           \simby
%         (\join{\bproc}{\aenv}) \punc{,} 
%           \label{eq:jiparamsim:def:jiparamsimby:jiparamsimequiv:jiparambisimby}
%         \\
%     \aproc \jiparamsimequiv{\aenv} \bproc
%       \; \funin & \: \Longleftrightarrow \;\;
%         \aproc \jiparamsimby{\aenv} \bproc
%           \;\logand\;
%         \bproc \jiparamsimby{\aenv} \aproc \punc{,} 
%           \label{eq:jiparamsimequiv:def:jiparamsimby:jiparamsimequiv:jiparambisimby} 
%         \\
%     \aproc \jiparambisim{\aenv} \bproc
%       \; \funin & \: \Longleftrightarrow \;\;
%         (\join{\aproc}{\aenv})
%           \bisim
%         (\join{\bproc}{\aenv}) \punc{.}
%           \label{eq:jiparambisim:def:jiparamsimby:jiparamsimequiv:jiparambisimby}
%   \end{align}  
% \end{defi}

%----------
\subsection*{Relationship of ji-parameterized bisimilarity with parameterized bisimilarity}
%----------

In order to recognize Larsen's parameterized bisimilarity as bisimilarity with respect to a specific form of \joininteraction, 
  we introduce a `right-determinizing' variant~$\sjoindot$ of the join operation~$\sjoin$.
    For interactions of a process~$\aproc$ with an environment~$\aenv$
      this operation yields the process~$\joindot{\aproc}{\aenv}$ 
        from which transitions are labeled by pairs $\pair{\aact}{\aenvacc}$ 
          that result from joining an \transitionact{\aact} from $\aproc$
            with an \transitionact{\aact} from $\aenv$ to target $\aenvacc$. 
              In this way different environment steps that originally have the same action label are distinguished
                from $\sjoindot$\nb-joins. 
Indeed, by making different targets of environment transitions visible as different transitions from $\joindot{\aproc}{\aenv}$ 
  for processes $\aproc$ and $\bproc$ and an environment $\aenv$,
    a correspondence arises between bisimulations that link $\joindot{\aproc}{\aenv}$ and $\joindot{\bproc}{\aenv}$ 
      % of processes $\aproc$ and $\bproc$ with an environment $\aenv$
       and parameterized bisimulations that link $\aproc$ and~$\bproc$~with~respect~to~$\aenv$.
    
% In order to model the specific form of \joininteraction\ on which Larsen's parameterized bisimilarity is based,
%   we introduce a `right-determinizing' variant~$\sjoindot$ of the join operation~$\sjoin$ 
%     that for interactions with environments 
%       records targets of environment steps in the arising action labels,
%         thereby distinguishing join interactions with different environment steps that (originally) have the same action label.
% This form of making different targets of environment transitions visible as different transitions of $\sjoindot$\nb-in\-ter\-ac\-tions 
%   will facilitate a correspondence between bisimulations on $\sjoindot$\nb-in\-ter\-ac\-tions $\joindot{\aproc}{\aenv}$ and $\joindot{\bproc}{\aenv}$ 
%     of processes $\aproc$ and $\bproc$ with an environment $\aenv$
%      and parameterized bisimulations that link $\aproc$ and $\bproc$ with respect to $\aenv$.
        
\begin{defi}[right-determinizing \joininteraction\ with environment \protect\LTSs]\label{def:joindot}
  Let
    $\aprocLTS = \tuple{\procs,\actions,\sprocltzero}$ be a process \LTS,
    and  $\aenvLTS = \tuple{\envs,\actions,\senvltzero}$ be an environment \LTS. 
  By the \emph{right-determinizing join-interaction\vspace{-2.5pt} of $\aprocLTS$ and $\aenvLTS$}
    we understand the \LTS\ of the form $\joindot{\aprocLTS}{\aenvLTS} = \tuple{ %\procsdot,
                                                                                 \joindot{\procs}{\envs}, \actions\times\envs, \sltzero }$
      with $%\procsdot 
            \joindot{\procs}{\envs} \defdby \descsetexpns{ \joindot{\aproc}{\aenv} }{ \aproc\in\procs,\, \aenv\in\envs }$ and  
      where $ \sltzero \subseteq %\procsdot
                                 (\joindot{\procs}{\envs})\times(\actions\times\envs)\times%\procsdot
                                                                                           (\joindot{\procs}{\envs}) $ is defined by 
        as transitions that are generated by the following rules:
      \begin{equation*}
        \AxiomC{$ \smash{\aproc \proclt{\aact} \aprocacc}\rule[-3pt]{0pt}{8pt} $}
        \AxiomC{$ \smash{\aenv  \envlt{\aact}  \aenvacc} $}
        \insertBetweenHyps{\hspace*{2em}}
        \BinaryInfC{$ \joindot{\aproc}{\aenv}  \lt{\pair{\aact}{\aenvacc}} \joindot{\aprocacc}{\aenvacc} $}
        \DisplayProof
      \end{equation*} 
      Hereby ``$\sjoindot$'' in processes $\joindot{\aproc}{\aenv}$ of %the \joininteraction\ \LTS\ 
                                                                        $\joindot{\aprocLTS}{\envs}$ 
        has to be understood as a term constructor that from any process $\aproc\in\procs$ and environment $\aenv\in\envs$
          constructs a formal \joininteraction\ process $\joindot{\aproc}{\aenv}$~in~$\joindot{\procs}{\envs}$. 
\end{defi}

% \begin{defi}[right-determinizing \joininteraction\ of \protect\LTSs]
%   We consider \LTSs\
%     $\aLTSi{1} = \tuple{\procsi{1},\actions,\sltzeroi{1}}$
%     and  $\aLTSi{2} = \tuple{\procsi{2},\actions,\sltzeroi{2}}$.
%     %
%   By the \emph{right-determinizing join(-interaction) of $\aLTSi{1}$ and $\aLTSi{2}$}
%     we understand the \LTS\ of the form $\joindot{\aLTSi{1}}{\aLTSi{2}} = \tuple{ \procsdot, \actions\times\procsi{2}, \sltzero }$
%       with $\procsdot \defdby \descsetexpns{ \joindot{\aproci{1}}{\aproci{2}} }{ \aproci{1}\in\procsi{1},\, \aproci{2}\in\procsi{2} }$ and  
%       where $ \sltzero \subseteq \procsdot\times(\actions\times\procsi{2})\times\procsdot $ is defined by 
%         the following transition system rule: 
%         % 
%       \begin{equation*}
%         \AxiomC{$ \aproci{1} \lti{\aact}{1} \aprocacci{1} $}
%         \AxiomC{$ \aproci{2} \lti{\aact}{2} \aprocacci{2} $}
%         \BinaryInfC{$ \joindot{\aproci{1}}{\aproci{2}}  \lt{\pair{\aact}{\aprocacci{2}}} \joindot{\aprocacci{1}}{\aprocacci{2}} $}
%         \DisplayProof
%       \end{equation*} 
%       %
%       The symbol ``$\sjoindot$'' in processes $\joindot{\aproc}{\bproc}$ of %the \joininteraction\ \LTS\ 
%                                                                       $\joindot{\aLTSi{1}}{\aLTSi{2}}$ 
%         is meant to be understood as a term constructor that from any two processes $\bproci{1}$ in $\procsi{1}$ and $\bproci{2}$ in $\procsi{2}$
%           constructs a formal \joininteraction\ process $\join{\aproci{1}}{\aproci{2}}$~in~$\procsdot$. 
% \end{defi}

Now this variant ``$\sjoindot$'' of the join operation ``$\sjoin$''
  facilitates characterizations of parameterized simulatability and bisimilarity
    that are analogous in kind to the definitions of \jiparameterized\ simulatability and bisimilarity in Definition~\ref{def:jiparamsimby:jiparamsimequiv:jiparambisimby}.
As stated by logical equivalences in the following lemma,
  parameterized simulatability, and parameterized bisimilarity
    correspond to simulatability, and respectively to bisimilarity,
      of $\sjoindot$\nb-in\-ter\-ac\-tions between two processes and an environment.
From this we obtain inclusions of parameterized simulatability and bisimilarity in \jiparameterized\ simulatability and bisimilarity.      
% As consequences we also obtain inclusions of $\sparamsimby{\aenv}$ in $\sjiparamsimby{\aenv}$ and of $\sparambisim{\aenv}$ in $\sjiparambisim{\aenv}$.       
                                    
\begin{lem}\label{lem:joindot:join}
  For all processes $\aproc$ and $\bproc$, and environments~$\aenv$ the following two chains of statements hold:
  \begin{alignat}{3}
    \aproc \paramsimby{\aenv} \bproc
      \;\; & \Longleftrightarrow\;\;
    (\joindot{\aproc}{\aenv}) 
      \simby
    (\joindot{\bproc}{\aenv}) 
    & 
    \qquad\qquad\quad &
    &
    \aproc \parambisim{\aenv} \bproc
      \;\; & \Longleftrightarrow\;\;
    (\joindot{\aproc}{\aenv}) 
      \bisim
    (\joindot{\bproc}{\aenv}) 
      \label{eq:1:lem:joindot:join}
    \\
      \;\; & \:\Longrightarrow\;\;
    (\join{\aproc}{\aenv}) 
      \simby
    (\join{\bproc}{\aenv})
    &
    &
    &
      \;\; & \:\Longrightarrow\;\;
    (\join{\aproc}{\aenv}) 
      \bisim
    (\join{\bproc}{\aenv})
      \label{eq:2:lem:joindot:join}
    \\
      \;\; & \:\Longleftrightarrow\;\;
    \aproc \jiparamsimby{\aenv} \bproc \punc{,}
    &
    &
    &
      \;\; & \:\Longleftrightarrow\;\;
    \aproc \jiparambisim{\aenv} \bproc \punc{,}
      \notag     
  \end{alignat}
  %
  % %
  % \begin{align}
  %   \aproc \parambisim{\aenv} \bproc
  %     \;\; & \Longleftrightarrow\;\;
  %   (\joindot{\aproc}{\aenv}) 
  %     \bisim
  %   (\joindot{\bproc}{\aenv}) 
  %     \label{eq:1:lem:joindot:join}
  %   \\
  %     \;\; & \:\Longrightarrow\;\;
  %   (\join{\aproc}{\aenv}) 
  %     \bisim
  %   (\join{\bproc}{\aenv})
  %     \label{eq:2:lem:joindot:join}
  %   \\
  %     \;\; & \:\Longleftrightarrow\;\;
  %   \aproc \jiparambisim{\aenv} \bproc \punc{,}
  %     \notag     
  % \end{align}
  % %
  where $\joindot{\aproc}{\aenv}$ and $\joindot{\bproc}{\aenv}$ 
    % on the right in \eqref{eq:1:lem:joindot:join}
      are processes from the right-determinizing \joininteraction~\LTS~$\joindot{\aprocLTS}{\aenvLTS}$,
        and $\join{\aproc}{\aenv}$ and $\join{\bproc}{\aenv}$
            are processes from the \joininteraction~\LTS~$\join{\aprocLTS}{\aenvLTS}$.
\end{lem}

\begin{proof}[Proof (Sketch).]
  We consider a process~\LTS~$\aprocLTS = \tuple{\procs,\actions,\sprocltzero}$, and an environment LTS~$\aenvLTS = \tuple{\envs,\actions,\senvltzero}$.
   
  We only argue for the chain of equivalences and implications on the right for $\sparambisim{\aenv}$, $\sbisim$, and $\sjiparambisim{\aenv}$, 
    since the chain of statements on the left for $\sparamsimby{\aenv}$, $\ssimby$, and $\sjiparamsimby{\aenv}$ can be demonstrated analogously.

  We first consider statement~\eqref{eq:1:lem:joindot:join}.
    For showing ``$\Rightarrow$'' 
      it suffices to demonstrate that if $\aparambisim = \family{\sabisimi{\aenv}}{\aenv\in\envs}$
        is an \parameterizedby{\aenvLTS} bisimulation on $\aprocLTS$, 
          then $\abisim \defdby \descsetexpns{ \pair{\joindot{\aproc}{\aenv}}{\joindot{\bproc}{\aenv}} }
                                             { \aproc \abisimi{\aenv} \bproc }$ is a bisimulation on $\joindot{\aprocLTS}{\aenvLTS}$.
    For ``$\Leftarrow$''
      it suffices to show that if $\abisim$ is a bisimulation on $\joindot{\aprocLTS}{\aenvLTS}$,
        then $\aparambisim = \family{\abisimi{\aenv}}
                     {\aenv\in\envs}$ 
             with the defining clause
               $\abisimi{\aenv} \defdby \descsetexpns{ \pair{\aproc}{\bproc} }{ \pair{\joindot{\aproc}{\aenv}}{\joindot{\bproc}{\aenv}} \in \abisim }$  
                 for $\aenv\in\envs$
        is an \parameterizedby{\aenvLTS} bisimulation on $\aprocLTS$. 
    Both auxiliary statements can be shown by using the conditions (forth) and (back) from the assumed (parameterized) bisimulation
      in order to demonstrate the conditions (forth) and (back) of the (parameterized) bisimulation in the conclusion of the implication.                                            
  
  The implication in \eqref{eq:2:lem:joindot:join}
    can be easily verified similarly:
      by showing that whenever $\abisimi{\bullet}$ is a bisimulation on $\joindot{\aprocLTS}{\aenvLTS}$,
        then $\abisim \defdby \descsetexpns{ \pair{\join{\aproc}{\aenv}}{\join{\bproc}{\aenv}} }
                                           { \pair{\joindot{\aproc}{\aenv}}{\joindot{\bproc}{\aenv}} \in \abisimi{\bullet} }$
        is a bisimulation on $\join{\aprocLTS}{\aenvLTS}$.
\end{proof}
\begin{figure}[t!]%
\begin{center}  
\begin{tikzpicture}

  \matrix[anchor=north,row sep=1cm,column sep=0.35cm,%every node/.style={draw,thick,circle,minimum width=2.5pt,fill,inner sep=0pt,outer sep=2pt},
            ampersand replacement=\& ] at (7.5,0) {
                   \&  \node(e-pos){}; \&                   \&[0.4cm] \&   \node(p-pos){};     \& \&[0.4cm]                \& \node(q-pos){}; \&               \& \&[0.75cm]     
                   \&[0.1cm] \node(pe-pos){}; \&[0.1cm]          \&[0.6cm]  \&                \& \&                \& \&[0.1cm] \& \node(qe-pos){};
    \\ 
    \node(e0-pos){};   \&              \&  \node(e1-pos){};     \&         \&   \node(p0-pos){};   \& \&        \node(q0-pos){};   \&             \&  \node(q1-pos){}; \& \&[0.75cm] 
    \node(pe0-pos){};  \&                     \&  \node(pe1-pos){};  \&[0.6cm]  \& \node(qe0-pos){};   \& \& \&  \node(qe1-pos){};  \& \&        \&[0.1cm] \& \node(qe2-pos){}; \& \& \& \node(qe3-pos){};              
    \\
    \node(e00-pos){};  \&              \&  \node(e10-pos){};    \&         \&   \node(p00-pos){};  \& \&        \node(q00-pos){};  \&             \&               \& \&[0.75cm]
    \node(pe00-pos){}; \&                     \&                 \&[0.6cm]  \& \node(qe00-pos){};
    \\  
  };
  %\draw[<-,very thick,>=latex,chocolate,shorten >=1pt](start) -- ++ (90:{0.5cm});
  %
  \path (e-pos) ++ (0cm,0cm) node[red](e){$\aenv$};
  \path (e0-pos) ++ (0cm,0cm) node[red](e0){$\bact$};
  \path (e00-pos) ++ (0cm,0cm) node(e00){$\deadlock$};
  \path (e1-pos) ++ (0cm,0cm) node(e1){$\deadlock$};
  \draw[-implies,double equal sign distance,firebrick](e) to node[left,pos=0.325]{$\aact$} (e0); 
  \draw[-implies,double equal sign distance](e) to node[right,pos=0.325]{$\aact$} (e1); 
  \draw[-implies,double equal sign distance](e0) to node[left,pos=0.4]{$\bact$} (e00); 
  \path (p-pos) ++ (0cm,0cm) node(p){$\aproc$};
  \path (p0-pos) ++ (0cm,0cm) node(p0){$\bact$};
  \path (p00-pos) ++ (0cm,0cm) node(p00){$\deadlock$};
  \draw[->,firebrick](p) to node[left,pos=0.4]{$\aact$} (p0); 
  \draw[->](p0) to node[left,pos=0.4]{$\bact$} (p00); 
  \path (q-pos) ++ (0cm,0cm) node(q){$\bproc$};
  \path (q0-pos) ++ (0cm,0cm) node(q0){$\bact$};
  \path (q00-pos) ++ (0cm,0cm) node(q00){$\deadlock$};
  \path (q1-pos) ++ (0cm,0cm) node(q1){$\deadlock$};
  \draw[->](q) to node[left,pos=0.325]{$\aact$} (q0); 
  \draw[->,firebrick](q) to node[right,pos=0.325]{$\aact$} (q1); 
  \draw[->](q0) to node[left,pos=0.4]{$\bact$} (q00);
  \draw[-,densely dashed,red,out=20,in=160,shorten >=0pt] 
    (p) to node{\scalebox{1.5}{$\times$}} node[pos=0.3,above]{$\aenv$} (q);
  \draw[-,densely dashed,red,out=35,in=150,shorten >=0pt] 
    (p0) to node[pos=0.4]{\scalebox{1.5}{$\times$}} node[pos=0.2,above]{$\bact$} (q1);

  \path(pe-pos) ++ (0cm,0cm) node(pe){$\joindot{\aproc}{\alert{\aenv}}$};
  \path(pe0-pos) ++ (0cm,0cm) node(pe0){\small $\joindot{\bact}{\alert{\bact}}$};
  \path(pe00-pos) ++ (0cm,0cm) node(pe00){\small $\joindot{\deadlock}{\deadlock}$};
  \path(pe1-pos) ++ (0cm,0cm) node(pe1){\small $ \joindot{\bact}{\deadlock}$};
  \draw[->,firebrick](pe) to node[left,pos=0.325]{\small $\pair{\aact}{\alert{\bact}}$} (pe0); 
  \draw[->](pe) to node[right,pos=0.325]{\small $\pair{\aact}{\deadlock}$} (pe1); 
  \draw[->](pe0) to node[left,pos=0.4,xshift=2pt]{\small $\pair{\bact}{\deadlock}$} (pe00); 
  \path(qe-pos) ++ (0cm,0cm) node(qe){$\joindot{\bproc}{\alert{\aenv}}$};
  \path(qe0-pos) ++ (0cm,0cm) node(qe0){\small $\joindot{\bact}{\bact}$};
  \path(qe00-pos) ++ (0cm,0cm) node(qe00){\small $\joindot{\deadlock}{\deadlock}$};
  \path(qe1-pos) ++ (0cm,0cm) node(qe1){\small $ \joindot{\deadlock}{\alert{\bact}}$};
  \path(qe2-pos) ++ (0cm,0cm) node(qe2){\small $ \joindot{\bact}{\deadlock}$};
  \path(qe3-pos) ++ (0cm,0cm) node(qe3){\small $ \joindot{\deadlock}{\deadlock}$};
  \draw[->](qe) to node[left,pos=0.325]{\small $\pair{\aact}{\bact}$} (qe0); 
  \draw[->,firebrick](qe) to node[pos=0.5,xshift=2pt]{\small $\pair{\aact}{\alert{\bact}}$} (qe1); 
  \draw[->](qe) to node[pos=0.5]{\small $\pair{\aact}{\deadlock}$} (qe2); 
  \draw[->](qe) to node[right,pos=0.325,xshift=1.5pt]{\small $\pair{\aact}{\deadlock}$} (qe3); 
  \draw[->,shorten >=-2pt](qe0) to node[left,pos=0.2,xshift=3.5pt]{\small $\pair{\bact}{\deadlock}$} (qe00);   
  \draw[-,red,densely dashed,out=20,in=160,shorten >=-2pt] 
    (pe) to node{\scalebox{1.5}{$\times$}} (qe);
  \draw[-,red,densely dashed,out=-30,in=220,shorten >=-2pt] 
    (pe0) to node[pos=0.4]{\scalebox{1.5}{$\times$}} (qe1);
%   %
%   \draw[-,thick,densely dashed,magenta,out=30,in=150,shorten >=-2pt] 
%     (pe0) to (qe0);
%   % 
%   \draw[-,thick,densely dashed,magenta,out=-40,in=225,shorten >=-2pt] 
%     (pe0) to (qe1); 
%   \draw[-,thick,densely dashed,magenta,out=-40,in=200,shorten >=-2pt,] 
%     (pe1) to (qe2); 
%   \draw[-,thick,densely dashed,magenta,out=-40,in=195,shorten >=-2pt] 
%     (pe1) to (qe3); 
%   %
%   \draw[-,thick,densely dashed,magenta,out=-10,in=190,shorten >=-2pt] 
%     (pe00) to (qe00);

\end{tikzpicture}  
\end{center}  
  \vspace*{-2ex}
  \caption{\label{fig:ex:lem:joindot:join}% 
           Example that witnesses the correspondence \eqref{eq:1:lem:joindot:join} in Lemma~\ref{lem:joindot:join}: 
           For the environment $\aenv \defdby \actpref{\aact}{\bact} \plus \aact$
             and the processes $\aproc \defdby \actpref{\aact}{\bact}$
                           and $\bproc \defdby \aenv$,
           it holds that $\aproc \notparambisim{\aenv} \bproc$
             (indicated by the \alert{mismatches} ${\bf \alert{{\times}}}$ when building a parameterized bisimulation on the left),
           and also $\joindot{\aproc}{\aenv} \notbisim \joindot{\bproc}{\aenv}$
               (indicated by the \alert{mismatches} ${\bf \alert{{\times}}}$ when building a bisimulation on the right).
           Note that in contrast $\aproc \jiparambisim{\aenv} \bproc$ holds $\aproc$, $\bproc$, $\aenv$, see Fig.~\ref{fig:lem:jiparamsim:subseteq:paramsim} later.
               }
\end{figure}%

In Figure~\ref{fig:ex:lem:joindot:join} we illustrate the characterization \eqref{eq:1:lem:joindot:join} of parameterized bisimilarity $\sparambisim{\aenv}$
  as bisimilarity of $\sjoindot$\nb-in\-ter\-ac\-tions by an example.
By using Lemma~\ref{lem:joindot:join}  %the characterization of $\sparambisim{\aenv}$ in \eqref{eq:1:lem:joindot:join}
  we can now show %, stated by the proposition below,
    that with respect to deterministic environments no difference arises between 
      parameterized bisimilarity and \jiparameterized\ bisimilarity.

\begin{prop}\label{prop:parambisim:jiparambisim:det:envs}
  $\parambisim{\aenv} \;\,=\;\, \jiparambisim{\aenv}\:$
          holds for all deterministic environments $\aenv$.
\end{prop}

\begin{proof}
  Let $\aprocLTS = \tuple{\procs,\actions,\sprocltzero}$ be a process LTS, 
    and $\aenvLTS = \tuple{\envs,\actions,\senvltzero}$ be an environment LTS.\vspace{-2.5pt} 
  The Proposition follows from Lemma~\ref{lem:joindot:join},
    once the converse implication ``$\Leftarrow$'' in \eqref{eq:2:lem:joindot:join},
      which is the only implication that is missing there for $\parambisim{\aenv}$ to coincide with $\jiparambisim{\aenv}\,$,  
        is shown to hold for deterministic environments~$\aenv\,$:
        \begin{equation}
          \text{$\aenv$ is deterministic}
            \;\;\Longrightarrow\;\;
              \bigl[\, 
                (\joindot{\aproc}{\aenv}) 
                  \bisim
                (\joindot{\bproc}{\aenv}) 
                \;\;\Longleftarrow\;\;
                (\join{\aproc}{\aenv}) 
                  \bisim
                (\join{\bproc}{\aenv})
              \,\bigr] \punc{.}
            \label{eq:prf:prop:parambisim:jiparambisim:det:envs}
        \end{equation}
        
    For this, it suffices to show,
      that whenever $\sabisim$ is a bisimulation on $\join{\aprocLTS}{\aenvLTS}$
        in which all environments that occur in joins in pairs in $\sabisim$ are deterministic,
          then $\sabisimi{\bullet} 
                  \defdby
                     \descsetexp{ \pair{ \joindot{\aproc}{\aenv} }{ \joindot{\bproc}{\aenv} } }
                                { \pair{ \join{\aproc}{\aenv} }{ \join{\bproc}{\aenv} } \in \sabisim }$
            is a bisimulation on $\joindot{\aprocLTS}{\aenvLTS}$.  
      The assumption that only deterministic environments occur in $\sabisim$
        is not too restrictive, because derivatives of deterministic environments are deterministic again.
        
    We let $\sabisim$ be a bisimulation on $\join{\aprocLTS}{\aenvLTS}$ as described, 
      and let $\sabisimi{\bullet}$ be defined as above.
      We have to show that $\sabisimi{\bullet}$ is a bisimulation on $\joindot{\aprocLTS}{\aenvLTS}$. 
      
    For showing\vspace{-2pt} the condition (forth) for $\sabisimi{\bullet}$ to be a bisimulation on $\joindot{\aprocLTS}{\aenvLTS}$,
      we let $\pair{ \joindot{\aproc}{\aenv} }{ \joindot{\bproc}{\aenv} } \in \sabisimi{\bullet}$,
        and a transition $\joindot{\aproc}{\aenv} \lt{l} \cproc$ be arbitrary, where $\cproc\in\joindot{\procs}{\envs}$, $l$ some label in $\actions\times\procs$.
          Due to operational semantics of $\joindot{\aprocLTS}{\aenvLTS}$,
            this transition must actually be of the form $\joindot{\aproc}{\aenv} \lt{\pair{\aact}{\aenvacc}} \joindot{\aprocacc}{\aenvacc}$,
              for $\aprocacc\in\procs$ and $\aenvacc\in\envs$.
      We have to show that there is $\dproc\in\joindot{\procs}{\envs}$
        such that $\joindot{\bproc}{\aenv} \lt{\pair{\aact}{\aenvacc}} \dproc$ 
               and $\pair{\cproc}{\dproc} = \pair{\joindot{\aprocacc}{\aenvacc}}{\dproc}\in\sabisimi{\bullet}$. 
      
    From $\joindot{\aproc}{\aenv} \lt{\pair{\aact}{\aenvacc}} \joindot{\aprocacc}{\aenvacc}$\vspace{-5pt}
      it follows that $\aproc \proclt{\aact} \aprocacc$ and $\aenv \envlt{\aact} \aenvacc$.  
    By the definition of $\join{\aprocLTS}{\aenvLTS}$ it follows that there is 
      also the transition $\join{\aproc}{\aenv} \lt{\aact} \join{\aprocacc}{\aenvacc}$ in $\join{\aprocLTS}{\aenvLTS}$.
    From $\pair{ \joindot{\aproc}{\aenv} }{ \joindot{\bproc}{\aenv} } \in \sabisimi{\bullet}$
      it follows by the definition of $\sabisimi{\bullet}$ that $\pair{ \join{\aproc}{\aenv} }{ \join{\bproc}{\aenv} } \in \sabisim$,
        and by the assumption on $\sabisim$ also that $\aenv$ is deterministic.
    Then it follows from the condition (forth) of $\sabisim$ as a bisimulation on $\join{\aprocLTS}{\aenvLTS}$
      that there is some $\dproc\in\join{\procs}{\envs}$ such that $\join{\bproc}{\aenv} \lt{\aact  } \dproc$ %$\join{\aprocacc}{\aenvdacc}$ 
        and $\pair{\join{\aprocacc}{\aenvacc}}{\dproc}\in\sabisim$. % {\pair{\aact}{\aenvdacc}}\in\sabisim$.  
    By the definition of $\join{\aprocLTS}{\aenvLTS}$
      we find that $\dproc = \join{\bprocacc}{\aenvdacc}$ and $\pair{\join{\aprocacc}{\aenvacc}}{\join{\bprocacc}{\aenvdacc}}\in\sabisim$   
        for some $\bprocacc\in\procs$ and $\aenvacc\in\envs$
          with $\bproc \proclt{\aact} \bprocacc$ and $\aenv \envlt{\aact} \aenvdacc$.
    But now\vspace{-4pt} from $\aenv \envlt{\aact} \aenvacc$ and $\aenv \envlt{\aact} \aenvdacc$ 
      we can conclude, because $\aenv$ is deterministic,
        that $\aenvdacc = \aenvacc$.
    From this we obtain $\pair{\join{\aprocacc}{\aenvacc}}{\join{\bprocacc}{\aenvacc}}\in\sabisim$,
      which entails $\pair{\joindot{\aprocacc}{\aenvacc}}{\joindot{\bprocacc}{\aenvacc}}\in\sabisimi{\bullet}$.
    From $\aenv \envlt{\aact} \aenvacc$ and $\aenv \envlt{\aact} \aenvacc$
      we also obtain $\joindot{\bproc}{\aenv} \lt{\pair{\aact}{\aenvacc}} \joindot{\bprocacc}{\aenvacc}$.  
    Therefore we have\vspace{-5pt}   found in $\dproc \defdby \joindot{\bprocacc}{\aenvacc}$ the desired  $\dproc\in\joindot{\procs}{\envs}$ 
      with $\joindot{\bproc}{\aenv} \lt{\pair{\aact}{\aenvacc}} \dproc$ 
       and $\pair{\cproc}{\dproc} = \pair{\joindot{\aprocacc}{\aenvacc}}{\dproc}\in\sabisimi{\bullet}$.        
     
    In this way we have established the condition (forth) for $\sabisimi{\bullet}$ to be a bisimulation on $\joindot{\aprocLTS}{\aenvLTS}$. 
      Since the condition (back) can be verified analogously, we conclude that $\sabisimi{\bullet}$ is indeed a bisimulation on $\joindot{\aprocLTS}{\aenvLTS}$. 
      
    By having shown that $\sabisimi{\bullet}$ is a bisimulation on $\joindot{\aprocLTS}{\aenvLTS}$
      under the assumption that $\abisim$ is a bisimulation on $\join{\aprocLTS}{\aenvLTS}$ in which only deterministic environments occur,
        we have established \eqref{eq:prf:prop:parambisim:jiparambisim:det:envs}, from which the proposition follows
          from Lemma~\ref{lem:joindot:join} as argued above. 
\end{proof}

The proposition below clarifies which inclusions hold in general between $\sparambisim{\aenv}$, $\sjiparambisim{\aenv}$, and $\sjiparamsimequiv{\aenv}$.

\begin{prop}\label{prop:parambisim:jiparambisim:jiparamsimequiv}
  The following set-theoretical relationships hold 
    between parameterized bisimilarity~$\sparambisim{\aenv}$,
            \jiparameterized\ bisimilarity~$\sjiparambisim{\aenv}$,
            and \jiparameterized\ simulation equivalence~$\sjiparamsimequiv{\aenv}\,$:%  (for $\aenv$ environments):
    % The following statements hold about the relationship
    %   between parameterized bisimilarity~$\sparambisim{\aenv}$,
    %           \jiparameterized\ bisimilarity~$\sjiparambisim{\aenv}$,
    %           and \jiparameterized\ simulation equivalence~$\sjiparamsimequiv{\aenv}$
    %   for environments~$\aenv$:
    %
  \begin{enumerate}[label={(\roman{*})},itemsep=0.15ex]
    \item{}\label{it:1:prop:parambisim:jiparambisim:jiparamsimequiv}
      $\parambisim{\aenv} \;\,\subseteq\;\, \jiparambisim{\aenv}\:$ for all environments~$\aenv$.
    \item{}\label{it:2:prop:parambisim:jiparambisim:jiparamsimequiv}  
      $\parambisim{\aenv} \;\,\neq\;\, \jiparambisim{\aenv}\:$
        for some environments~$\aenv$, 
          for which then $\,\parambisim{\aenv} \;\,\subsetneqq\;\, \jiparambisim{\aenv}\:$ holds due to \ref{it:1:prop:parambisim:jiparambisim:jiparamsimequiv}.
        % For some environments~$\aenv$
        % it holds that $\,\jiparambisim{\aenv} \;\,\not\subseteq\;\, \parambisim{\aenv}\,$,
        %   and hence $\,\parambisim{\aenv} \;\,\neq\;\, \jiparambisim{\aenv}\,$,
        %         and $\,\parambisim{\aenv} \;\,\subsetneqq\;\, \jiparambisim{\aenv}\,$ by \ref{it:1:prop:parambisim:jiparambisim:jiparamsimequiv}.
        %
    \item{}\label{it:3:prop:parambisim:jiparambisim:jiparamsimequiv}             
      $\jiparambisim{\aenv} \;\,\subseteq\;\, \jiparamsimequiv{\aenv}\:$ for all environments~$\aenv$.
    \item{}\label{it:4:prop:parambisim:jiparambisim:jiparamsimequiv}             
      $\jiparambisim{\aenv} \;\,\neq\;\, \jiparamsimequiv{\aenv}\:$ for some environments~$\aenv$,
          for which then $\,\jiparambisim{\aenv} \;\,\subsetneqq\;\, \jiparamsimequiv{\aenv}\:$ holds due to \ref{it:3:prop:parambisim:jiparambisim:jiparamsimequiv}.
  \end{enumerate}       
  The counterexample statements \ref{it:2:prop:parambisim:jiparambisim:jiparamsimequiv} and \ref{it:4:prop:parambisim:jiparambisim:jiparamsimequiv}
    hold under the proviso
      that environments are included among processes, they permit at least two actions, and are closed under action prefixing and sums.
  This can be weakened to merely require that $\actpref{\aact}{\bact} \plus \aact$ and $\actpref{\aact}{\bact}$ 
  are contained among environments and processes.
\end{prop}  
\input{fig-ex1.tex}\vspace*{-2ex}%
\begin{proof}
  Statement~\ref{it:1:prop:parambisim:jiparambisim:jiparamsimequiv} 
    follows directly from the chain of implications as guaranteed by Lemma~\ref{lem:joindot:join}.
    %
    % In order to show \ref{it:1:prop:parambisim:jiparambisim:jiparamsimequiv} 
    %   we argue as follows, for all environments~$\aenv$ and processes $\aproc$ and $\bproc\,$:
    %   %
    % \begin{alignat}{2}
    %   \aproc \parambisim{\aenv} \bproc
    %     %
    %     \;\; & \Longleftrightarrow\;\;
    %     %
    %   (\joindot{\aproc}{\aenv}) 
    %     \bisim
    %   (\joindot{\bproc}{\aenv})
    %     & \qquad & 
    %       \text{(by \eqref{eq:1:lem:joindot:join} in Lemma~\ref{lem:joindot:join})}
    %   % 
    %   \\
    %     \;\; & \:\Longrightarrow\;\;
    %     %
    %   \join{\aproc}{\aenv}
    %     \bisim
    %   \join{\bproc}{\aenv} 
    %     & &
    %       \text{(by \eqref{eq:2:lem:joindot:join} in Lemma~\ref{lem:joindot:join})}
    %   %
    %   \\
    %     \;\; & \Longleftrightarrow\;\;
    %     %
    %     \aproc \jiparambisim{\aenv} \bproc 
    %     & &
    %       \text{(by definition of $\sjiparambisim{\aenv}$).}
    %     %
    % \end{alignat}
  
  A counterexample for \ref{it:2:prop:parambisim:jiparambisim:jiparamsimequiv} 
    is in Figure~\ref{fig:1}: 
%    $\aenv \defdby \actpref{a}{b} + a$.
%  To see this, let $\aproc \defdby \actpref{a}{b}$ and $\bproc \defdby \aenv$.
  We have that $\aproc \jiparambisim{\aenv} \bproc$ holds
       due to $ \join{\aproc}{\aenv} =      \join{(\actpref{a}{b} + a)}{(\actpref{a}{b})}
                                     \iso   \actpref{a}{b} + a 
                                     \bisim \actpref{a}{b + a + a + a}
                                     \iso   \join{(\actpref{a}{b} + a)}{(\actpref{a}{b} + a)}
                                     =      \join{\bproc}{\aenv} $,
       where $\siso$ denotes being isomorphic,
    which shows $ \join{\aproc}{\aenv} \bisim \join{\bproc}{\aenv}$.  
  However, $\aproc \notparambisim{\aenv} \bproc$ holds for the following reason:
    Suppose that $\aproc \parambisim{\aenv} \bproc$ holds.
      Then due to $\aenv = \actpref{a}{b} + a \envlt{a} b$,\vspace{-4.25pt}
        and the condition (back) of an underlying parameterized bisimilarity
          the transition $\bproc \proclt{a} 0$ must be matched by the transition $\aproc \proclt{a} b$
            so that $b \parambisim{b} 0$ holds.
              However the latter is false, 
            because $b \notparambisim{b} 0$ holds, as the environment $b$ and the process $b$ can make a $b$\nb-step, but $0$ cannot.
Statement \ref{it:3:prop:parambisim:jiparambisim:jiparamsimequiv}
follows from the fact that bisimilarity is symmetric and it is a 
simulation \cite{sang:2011}. For \ref{it:4:prop:parambisim:jiparambisim:jiparamsimequiv}, let $\aproc$ and
$\bproc$ be as in Figure~\ref{fig:1}, and let $\aenv = \aproc$. Then $\aproc \jiparamsimequiv{\aenv} \bproc$ holds due to 
$\join{\aproc}{\aenv} = \join{(\actpref{a}{b})}{(\actpref{a}{b})}
\iso \actpref{a}{b} \simbyequiv \actpref{a}{b} + a \iso 
\join{(\actpref{a}{b} + a)}{(\actpref{a}{b})} = \join{\bproc}{\aenv}$. However,
$\aproc \notjiparambisim{\aenv} \bproc\,$: indeed, $\join{\bproc}{\aenv} \proclt{a} \join{b}{0} \bisim 0$, to which 
$\join{\aproc}{\aenv}$ can only answer by reducing to $\join{b}{b}
\bisim b$, and clearly $0 \not\bisim b$.
\end{proof}

\subsection*{Parameterized simulatability coincides with \protect\jiparameterized\ simulatability}
  %\subsection*{JI-Parameterized similarity = parameterized similarity}
%----------

While parameterized bisimilarity and \jiparameterized\ bisimilarity 
  are two different relations in general by Proposition~\ref{prop:parambisim:jiparambisim:jiparamsimequiv}, \ref{it:2:prop:parambisim:jiparambisim:jiparamsimequiv},  
    it turns out that this does not hold for the corresponding two concepts of parameterized simulatability. 
    % Now in contrast to \jiorparameterized\ bisimilarity,
    %   which due to Proposition~\ref{prop:parambisim:jiparambisim:jiparamsimequiv}, \ref{it:2:prop:parambisim:jiparambisim:jiparamsimequiv},  
    %     are two different concepts in general,
    % this is not the case for \jiorparameterized\ simulatability.
For us it was surprising to find the proof of the first of the following two lemmas,
  which together show that
    parameterized simulatability and \jiparameterized\ simulatability coincide.

% For us somewhat surprisingly at first it turned out
% , guaranteed by the two lemmas below
%   (and specifically by the more important first one) and formulated by the subsequent proposition, 
%     that in contrast to \jiorparameterized\ bisimilarity,
%       which due to Proposition~\ref{prop:parambisim:jiparambisim:jiparamsimequiv}, \ref{it:2:prop:parambisim:jiparambisim:jiparamsimequiv}, 
%         are different concepts in general,
%           \jiparameterized\ simulatability and parameterized simulatability coincide.

\begin{lem}\label{lem:jiparamsim:subseteq:paramsim}
  $\: \jiparamsimby{\aenv} \;\,\subseteq\;\, \paramsimby{\aenv} \;$ holds for all environments~$\aenv$.
\end{lem}

\begin{proof}
  We fix a process~\LTS~$\aprocLTS = \tuple{\procs,\actions,\sprocltzero}$, 
    and an environment~\LTS~$\aenvLTS = \tuple{\envs,\actions,\senvltzero}$.\vspace{-1.5pt}
    
  As the crucial stepping stone, we show that $\aparamsim = \familynormalsize{ \asimi{\aenv} }{\aenv\in\envs}$ as defined by,
    for all $\aenv\in\envs\,$:
  \begin{equation}\label{eq:1:lem:jiparamsim:subseteq:paramsim}
    \asimi{\aenv} 
      \,\defdby\,
        \descsetexpbig{ \pair{\aproc}{\bproc} \in \procs }
                      {  \existsstzero{\aenvi{2}\in\envs} \bigl[ \join{\aproc}{\aenv} \:\simby\: \join{\bproc}{\aenvi{2}} \bigr] } 
          \subseteq \procs\times\procs 
  \end{equation}
  is an \parameterizedby{\aenvLTS} simulation on $\aprocLTS$.
  For this, we let $\aenv\in\envs$, and $\pair{\aproc}{\bproc} \in \sasimi{\aenv}$ be arbitrary.
    We assume that $\aenv \envlt{\aact} \aenvacc$, and $\aproc \proclt{\aact} \aprocacc$ for some $\aact\in\actions$, $\aenvacc\in\envs$, and $\aprocacc\in\procs$.%
      \vspace*{-4.25pt}
      We have to show that there exists $\bprocacc\in\procs$ with $\bproc \proclt{\aact} \bprocacc$ such that $\pair{\aprocacc}{\bprocacc}\in\sasimi{\aenvacc}$.
      
  From $\aenv \envlt{\aact} \aenvacc$ and $\aproc \proclt{\aact} \aprocacc$ we find that
    $\join{\aproc}{\aenv} \lt{\aact} \join{\aprocacc}{\aenvacc}\,$ holds.\vspace{-2pt}   
  Due to $\pair{\aproc}{\bproc} \in \sasimi{\aenv}$ we can pick $\aenvi{2}\in\envs$
    with $\join{\aproc}{\aenv} \simby \join{\bproc}{\aenvi{2}}$. 
  It follows, by the forward-property (forth) of the (largest) simulation~$\ssimby$
              applied to $\join{\aproc}{\aenv} \simby \join{\bproc}{\aenvi{2}}$ and $\join{\aproc}{\aenv} \lt{\aact} \join{\aprocacc}{\aenvacc}$,
              and by the operational semantics of the join operation,
    that there are $\bprocacc\in\procs$ and $\aenvacci{2}\in\envs$ such that\vspace{-4.25pt}
      $\join{\bproc}{\aenvi{2}} \proclt{\aact} \join{\bprocacc}{\aenvacci{2}}$,
      as well as $\bproc \proclt{\aact} \bprocacc$ and $\aenvi{2} \envlt{\aact} \aenvacci{2}$
      and with $\join{\aprocacc}{\aenvacc} \simby \join{\bprocacc}{\aenvacci{2}}$. 
        The latter shows that \mbox{$\pair{\aprocacc}{\bprocacc}\in\sasimi{\aenvacc}$},  
          and thus we have found $\bproc \proclt{\aact} \bprocacc$ such that $\pair{\aprocacc}{\bprocacc}\in\sasimi{\aenvacc}$.
  In this way we have verified that $\aparamsim = \familynormalsize{ \asimi{\aenv} }{\aenv\in\envs}$ as defined in \eqref{eq:1:lem:jiparamsim:subseteq:paramsim}
    is an \parameterizedby{\aenvLTS} simulation~on~$\aprocLTS$.
  
  \smallskip
  For showing $\: \jiparamsimby{\aenv} \;\,\subseteq\;\, \paramsimby{\aenv} \;$, 
    suppose now that $\aproc \jiparamsimby{\aenv} \bproc$ holds, for some $\aproc,\bproc\in\procs$ and $\aenv\in\envs$. 
      By the definition of $\sjiparamsimby{\aenv}$, this means that $\join{\aproc}{\aenv} \simby \join{\bproc}{\aenv}$ holds.
        That, however, implies $\pair{\aproc}{\bproc} \in \sasimi{\aenv}$ due to \eqref{eq:1:lem:jiparamsim:subseteq:paramsim}.
          But since we have recognized $\aparamsim$ as an \parameterizedby{\aenvLTS} simulation,
            we conclude that $\aproc \paramsimby{\aenv} \bproc$ holds.
\end{proof}

\begin{lem}\label{lem:paramsim:subseteq:jiparamsim}
  $\: \paramsimby{\aenv} \;\,\subseteq\;\, \jiparamsimby{\aenv} \;$ holds for all environments~$\aenv$.
\end{lem}

\begin{proof}%[Sketch of the Proof]
  The inclusion as stated by the lemma follows from the chain of implications displayed on the left in Lemma~\ref{lem:joindot:join},
    which as stated in its proof can be proved analogously as the implications on the right there.
  But since we dropped the argument there, we also provide the sketch of a direct proof here.
      
  % The inclusion in the lemma is guaranteed by the chain of implications displayed on the right within Lemma~\ref{lem:joindot:join}.
  % %
  % Alternatively, it can be established directly as follows.
  
  Let $\aparamsim = \family{\asimi{\aenv}}{\aenv\in\envs}$ be an \parameterizedby{\aenvLTS} simulation 
    on a process~\LTS $\aprocLTS = \tuple{\procs,\actions,\sprocltzero}$ 
      with respect to an environment~\LTS\ $\aenvLTS = \tuple{ \envs, \actions, \senvltzero }$.
  Then it is easy to verify that:
  \begin{equation*}
     \sasim \defdby \descsetexpns{ \tuple{ \join{\aproc}{\aenv}, \join{\bproc}{\aenv} } }
                                 { \aproc, \bproc \in \procs \text{ and } \aenv\in\envs \text{ such that } \aproc \asimi{\aenv} \bproc } 
            \:\subseteq\: \join{\procs}{\envs}                  
  \end{equation*}
  is a simulation on $\join{\aprocLTS}{\aenvLTS} = \tuple{ \join{\procs}{\envs}, \actions, \sltzero }$. 
  This statement implies that if $\aproc \paramsimby{\aenv} \bproc$,
    then $\join{\aproc}{\aenv} \simby \join{\bproc}{\aenv}$ follows, 
      and hence $\aproc \jiparamsimby{\aenv} \bproc$.
\end{proof}

% \begin{lem}\label{lem:jiparamsim:paramsim}
%   For all environments~$\aenv$ the following inclusions statements hold:
%   %
%   \begin{enumerate}[label={(\roman{*})},itemsep=0ex]
%     \item{}\label{lem:jiparamsim:paramsim}
%       $\: \paramsimby{\aenv} \;\,\subseteq\;\, \jiparambisim{\aenv} \,$.
%     \item{}\label{lem:jiparamsim:paramsim}
%       $\: \jiparamsimby{\aenv} \;\,\subseteq\;\, \parambisim{\aenv} \,$.
%   \end{enumerate}
% \end{lem}

\begin{prop}\label{prop:jiparamsim:equals:parasim}
    % JI-Parameterized similarity coincides with parameterized similarity,
    %   and \jiparameterized\ simulation equivalence coincides with parameterized simulation equivalence. 
  For all environments~$\aenv$ the following two statements hold:
    % , showing that $\jiparamsimby{\aenv}$ equals $\paramsimby{\aenv}$ 
    % and that $\jiparamsimequiv{\aenv}$ equals $\paramsimequiv{\aenv}$ for all environments: 
        %
  \begin{enumerate}[label={(\roman{*})},align=right,itemsep=0ex]
    \item{}\label{it:1:prop:jiparamsim:equals:parasim}
      $\:\jiparamsimby{\aenv} \;\,=\;\, \paramsimby{\aenv}\:$.
        %for all environments $\aenv$.
    \item{}\label{it:2:prop:jiparamsim:equals:parasim}
      $\:\jiparamsimequiv{\aenv} \;\,=\;\, \paramsimequiv{\aenv}\:$.
        %for all environments $\aenv$.
  \end{enumerate}  
\end{prop}

\begin{proof}
  The inclusions ``$\subseteq$'' and ``$\supseteq$'' that make up statement~\ref{it:1:prop:jiparamsim:equals:parasim}
    are guaranteed
      by Lemma~\ref{lem:jiparamsim:subseteq:paramsim} and by Lemma~\ref{lem:paramsim:subseteq:jiparamsim}, respectively. %, for all environments~$\aenv$.
  Then %statement~
       \ref{it:2:prop:jiparamsim:equals:parasim} follows from %statement~
                                                              \ref{it:1:prop:jiparamsim:equals:parasim}     
    by:
  \mbox{%  
  $\jiparamsimequiv{\aenv} \,=\,  \jiparamsimby{\aenv} \cap \canjiparamsim{\aenv}
                           \,=\,  \paramsimby{\aenv}   \cap \canparamsim{\aenv}
                           \,=\,  \paramsimequiv{\aenv} \,$.}
\end{proof}

The theorem below collects results we have obtained about which inclusions hold in general %, and which not,
  between parameterized bisimilarity, \jiparameterized\ bisimilarity, and \jiorparameterized\ simulation equivalence.

\begin{thm}\label{thm:incl:jiorparam:bisim:sim}
  The following set-theoretical relationships hold
    between parameterized bisimilarity, \jiparameterized\ bisimilarity, and \jiorparameterized\ simulation equivalence,
      for environments $\aenv, \benv, \cenv$:
      \begin{equation*}
        \begin{matrix}%{7}{ccccccc}
          \sparambisim{\aenv}
            & 
            \underset{\text{\nf Prop.$\,$\ref{prop:parambisim:jiparambisim:jiparamsimequiv},\ref{it:1:prop:parambisim:jiparambisim:jiparamsimequiv}}}{\subseteq} 
            &
          \sjiparambisim{\aenv} 
            &
            \,\underset{\text{\nf Prop.$\,$\ref{prop:parambisim:jiparambisim:jiparamsimequiv},\ref{it:3:prop:parambisim:jiparambisim:jiparamsimequiv}}}{\subseteq}\,
            & 
          \sparamsimequiv{\aenv} 
            &
            \underset{\text{\nf Prop.$\,$\ref{prop:jiparamsim:equals:parasim},\ref{it:2:prop:jiparamsim:equals:parasim}}}{=} 
            &
          \sjiparamsimequiv{\aenv}
          &
          \qquad \text{(for \underline{\smash{all}} $\aenv$)} \punc{,}
          \\[3ex]
          \sparambisim{\benv}
            & 
            \underset{\text{\nf Prop.$\,$\ref{prop:parambisim:jiparambisim:jiparamsimequiv},\ref{it:2:prop:parambisim:jiparambisim:jiparamsimequiv}}}{\subsetneqq}
            & 
          \sjiparambisim{\benv} 
            &
            \underset{\text{\nf Prop.$\,$\ref{prop:parambisim:jiparambisim:jiparamsimequiv},\ref{it:3:prop:parambisim:jiparambisim:jiparamsimequiv}}}{\subseteq} 
            &
          \sparamsimequiv{\benv} 
            &
            \underset{\text{\nf Prop.$\,$\ref{prop:jiparamsim:equals:parasim},\ref{it:2:prop:jiparamsim:equals:parasim}}}{=} 
            &
          \sjiparamsimequiv{\benv} 
          &
          \qquad \text{(for \underline{\smash{some}}~$\benv$}) \punc{,}
          \\[3ex]
          \sparambisim{\cenv}
            & 
            \underset{\text{\nf Prop.$\,$\ref{prop:parambisim:jiparambisim:jiparamsimequiv},\ref{it:1:prop:parambisim:jiparambisim:jiparamsimequiv}}}{\subseteq}
            & 
          \sjiparambisim{\cenv} 
            &
            \underset{\text{\nf Prop.$\,$\ref{prop:parambisim:jiparambisim:jiparamsimequiv},\ref{it:4:prop:parambisim:jiparambisim:jiparamsimequiv}}}{\subsetneqq} 
            &
          \sparamsimequiv{\cenv} 
            &
            \underset{\text{\nf Prop.$\,$\ref{prop:jiparamsim:equals:parasim},\ref{it:2:prop:jiparamsim:equals:parasim}}}{=} 
            &
          \sjiparamsimequiv{\cenv} 
          &
          \qquad \text{(for \underline{\smash{some}}~$\cenv$}) \punc{,}
        \end{matrix}
      \end{equation*}      
      % \begin{equation*}
      %   \begin{matrix}%{7}{ccccccc}
      %     \sparambisim{\aenv}
      %       & 
      %       \underset{\text{\nf Prop.$\,$\ref{prop:parambisim:jiparambisim:jiparamsimequiv},\ref{it:1:prop:parambisim:jiparambisim:jiparamsimequiv}}}{\subseteq} 
      %       &
      %     \sjiparambisim{\aenv} 
      %       &
      %       \,\underset{\text{\nf Prop.$\,$\ref{prop:parambisim:jiparambisim:jiparamsimequiv},\ref{it:3:prop:parambisim:jiparambisim:jiparamsimequiv}}}{\subseteq}\,
      %       & 
      %     \sparamsimequiv{\aenv} 
      %       &
      %       \underset{\text{\nf Prop.$\,$\ref{prop:jiparamsim:equals:parasim},\ref{it:2:prop:jiparamsim:equals:parasim}}}{=} 
      %       &
      %     \sjiparamsimequiv{\aenv}
      %     &
      %     \qquad \text{(for \underline{\smash{all}} $\aenv$)} \punc{,}
      %     %
      %     \\[3ex]
      %     %
      %     \sparambisim{\benv}
      %       & 
      %       \underset{\text{\nf Prop.$\,$\ref{prop:parambisim:jiparambisim:jiparamsimequiv},\ref{it:2:prop:parambisim:jiparambisim:jiparamsimequiv}}}{\subsetneqq}
      %       & 
      %     \sjiparambisim{\benv} 
      %       &
      %       \underset{\text{\nf Prop.$\,$\ref{prop:parambisim:jiparambisim:jiparamsimequiv},\ref{it:4:prop:parambisim:jiparambisim:jiparamsimequiv}}}{\subsetneqq} 
      %       &
      %     \sparamsimequiv{\benv} 
      %       &
      %       =
      %       % \underset{\text{\nf Prop.$\,$\ref{prop:jiparamsim:equals:parasim},\ref{it:2:prop:jiparamsim:equals:parasim}}}{=} 
      %       &
      %     \sjiparamsimequiv{\benv} 
      %     &
      %     \qquad \text{(for \underline{\smash{some}}~$\benv$}) \punc{,}
      %     %
      %   \end{matrix}
      % \end{equation*}
      %
      where %$\aenv, \benv, \cenv$ are environments,
        %and 
        the statements that guarantee the relationship in question are indicated.
\end{thm}

%----------
\subsection*{Discrimination \protect\preorder\ induced by (ji-)parameterized similarity}
%----------

Larsen noted in \cite[p.$\,$209--210]{lars:1987}:
  ``Due to the modal characterization [see Theorem~\ref{thm:log-char:parambisim}]
      and the simple characterization of the discrimination ordering presented [see Theorem~\ref{thm:char:discr-preorder:parambisim}], 
        we are  confident that the notion of parameterized bisimulation equivalence proposed is indeed a natural one''.
Indeed Larsen also explains that 
  ``the simulation ordering does not characterize the discrimination ordering generated by this alternative parameterized version [namely $\sjiparambisim{\aenv}$]''.
  
This is witnessed by the following proposition.
  Indeed it demonstrates
    that a characterization 
      of the discrimination \preorder\ induced by \jiparameterized\ bisimilarity~$\sjiparambisim{\aenv}$
        cannot, in analogy with Theorem~\ref{thm:char:discr-preorder:parambisim} for~$\sparambisim{\aenv}\,$,
          be of the form 
            $ \:
              \aenv \simby \benv 
                \;\Longleftrightarrow\;  
              \sjiparambisim{\benv}
               \;\subseteq\;
             \sjiparambisim{\aenv} \: $,
        for all environments~$\aenv$ and $\benv$.

\begin{prop}\label{prop:discr-preorder:jiparambisim}
  There are environments~$\aenv$ and $\benv$ such that:
  \begin{equation}
    \aenv \simby \benv
      \;\; \logand \;\;\;
     \sjiparambisim{\benv}
       \;\not\subseteq\;
     \sjiparambisim{\aenv} \:\punc{.}
  \end{equation}
\end{prop}

\begin{proof}
  Let $\aenv = \actpref{\aact}{\bact}$ and 
  $\benv = \actpref{\aact}{\bact} \plus \aact$. We have that 
  $\aenv \simby \benv$. Set $\aproc = \aenv$ and 
  $\bproc = \benv$. We have that $\aproc \jiparambisim{\benv}
  \bproc$ but $\aproc \notjiparambisim{\aenv}
  \bproc$, hence $\sjiparambisim{\benv}$ is not contained
  in $\sjiparambisim{\aenv}$.   
\end{proof}

Unfortunately we have not yet found an appealing characterization of the discrimination \preorder\ that is induced by $\sjiparambisim{\aenv}$.
We formulate this question together with a perhaps also interesting specialization
  as the open problems \ref{P1} and \ref{P2} in the conclusion.

% \noindent
% \todo{refer to the following open problem in the conclusion}
% %
% \begin{oprob}
%   Does equality of \jiparameterized\ bisimilarity with respect to environments~$\aenv$ and~$\benv$
%     coincide with bisimilarity of $\aenv$ and~$\benv\,$?
%       Equivalently, does the implication ``$\Leftarrow$'' hold in the following statement
%         (of which ``$\Rightarrow$'' is easy to verify), for all environments~$\aenv$ and $\benv\,$:
%       %
%   \begin{equation}
%     \aenv \bisim \benv
%       \;\;\; \overset{\displaystyle \overset{?}{\Longleftarrow}}
%                      {\:\:\Longrightarrow} \;\;\;
%     \jiparambisim{\aenv} \; =\; \jiparambisim{\benv} \punc{.}                 
%   \end{equation}
% \end{oprob}

However, and somewhat surprisingly,
    we do obtain characterizations analogous to Theorem~\ref{thm:char:discr-preorder:parambisim}
      for the discrimination \preorder\ on environments $\aenv$ with respect to \jiorparameterized\ similarity $\sjiparamsimby{\aenv}$ and~$\sparamsimby{\aenv}$,
        and with respect to \jiorparameterized\ simulation equivalence $\sjiparamsimequiv{\aenv}$ and~$\sparamsimequiv{\aenv}$.

Similar to the proviso for the `completeness' direction ``$\Rightarrow$'' 
  in \eqref{eq:thm:char:discr-preorder:parambisim} of Theorem~\ref{thm:char:discr-preorder:parambisim}
    our characterization of $\sjiparamsimequiv{\aenv}$ 
      requires an assumption that guarantees that the structure of the underlying process \LTS\
        is sufficiently rich in relation to the environment \LTS.
While Larsen assumed closure under the formation of action prefixing and finite sums,
  we will assume the existence of a `universal' process,
    and that the underlying process \LTS\ contains the environment \LTS, and is closed under the formation of joins.
(The assumption of a universal process simplifies the proof, but can be dropped.)  
    
Let $\aprocLTS = \tuple{\procs,\actions,\sprocltzero}$ be a process \LTS.
    % If a process $\univproc\in\actions$ of $\aprocLTS$ 
    %   has the property 
    %     that for all $\aact\in\actions$ there is a transition $\univproc \proclt{\aact} \univproc$,
    %       then we call $\univproc$ \emph{the canonical universal process in $\aprocLTS$}.
We say that a process $\aunivproc\in\procs$
  is \emph{universal} 
    if for all $\aact\in\actions$ there is 
      a transition $\aunivproc \proclt{\aact} \aunivprocacc$ with $\aunivprocacc\in\procs$ in $\aprocLTS$
        such that $\aunivprocacc \bisim \aunivproc$
          (consequently it permits all actions in transitions from any of its reachable states).
Note that all universal processes in $\aprocLTS$ are bisimilar. %,
    % and that if $\aprocLTS$ contains the canonical universal process~$\univproc$,
    % then $\aunivproc \bisim \univproc$ for all universal processes $\aunivproc\in\procs$. 
If $\aprocLTS$ is additionally closed under the formation~$\sjoin$,
  then $\join{\aunivproc}{\aproc} \bisim \join{\aproc}{\aunivproc} \bisim \aproc$
    holds for all $\aproc,\aunivproc\in\procs$ where $\aunivproc$ is universal.           
    
For proving our characterization below, Theorem~\ref{thm:char:discr-preorder:jiparamsim:jiparamsimequiv}
  we will use the following lemma.    
    
\begin{lem}\label{lem:aux1}
For all environments $\aenv,\benv$ it holds,
  provided that the environment \LTS\ is closed under joins:
\[
\aenv \simby \join{\benv}{\aenv} \;\;\iff\;\; \aenv \simby \benv \punc{.}
\]
\end{lem}
\begin{proof}
  For showing the direction ``$\Rightarrow$'', and the direction ``$\Leftarrow$'' in the statement of the lemma,
    it suffices to prove that the relation
        $\{\pair{\aenv}{\benv} \mid \aenv \simby \join{\benv}{\aenv}\}$, 
        and respectively,
          that the relation $\{\pair{\aenv}{\join{\benv}{\aenv}} \mid \aenv \simby \benv\}$
      is a simulation. Both of these statements can be verified in a straightforward manner. 
\end{proof}

% \isnew{A process $\aproc$ is \emph{universal} iff for all 
% $\aact\in\actions$ it holds that $\aproc \proclt{\aact} \bproc$ for
% some $\bproc$, and moreover for all $\bproc$ if 
% $\aproc \proclt{\aact} \bproc$ then $\bproc$ is universal. Notice 
% that every universal process $\aproc$ is the unity of $\join{}{}$
% up-to bisimilarity,
% that is $\join{\aproc}{\bproc} \bisim \join{\bproc}{\aproc} \bisim
%  \bproc$
% for all $\bproc$. In Theorem~\ref{thm:char:discr-preorder:jiparamsim:jiparamsimequiv} below we assume that the set of
% processes contains at least a universal process and all the 
% environments $\aenv$.}

We now are in a position to show characterizations of the discrimination \preorder{s} 
  of \jiorparameterized\ simulatability and \jiorparameterized\ simulation equivalence,
    as formulated by the theorem below.

\begin{thm}%[analogues to Theorem~\ref{thm:char:discr-preorder:parambisim} for $\sjiparamsimby{\cenv}$ and $\sjiparamsimequiv{\cenv}$] 
      \label{thm:char:discr-preorder:jiparamsim:jiparamsimequiv}
  The following logical equivalences hold,    
    for all environments $\aenv$ and $\benv$, 
      provided that: the underlying process \LTS\ 
        contains a universal process and the environment \LTS,
        and additionally
        is closed under the formation of joins
        \emph{(but note that image-finiteness as in Theorem~\ref{thm:char:discr-preorder:parambisim} is not required)}:
  \begin{align}
    \aenv \simby \benv
      \;\; & \Longleftrightarrow\;\;\;
         \sjiparamsimby{\benv}
           \:\subseteq\:
         \sjiparamsimby{\aenv} \:\punc{,} 
         \label{eq:1:thm:char:discr-preorder:jiparamsim:jiparamsimequiv}     
         \displaybreak[0]\\
    \aenv \simby \benv
      \;\; & \Longleftrightarrow\;\;\;
         \scanjiparamsim{\benv}
           \:\subseteq\:
         \scanjiparamsim{\aenv} \:\punc{,} 
         \label{eq:2:thm:char:discr-preorder:jiparamsim:jiparamsimequiv}     
         \displaybreak[0]\\
    \aenv \simby \benv
      \;\; & \Longleftrightarrow\;\;\;
         \sjiparamsimequiv{\benv}
           \:\subseteq\:
         \sjiparamsimequiv{\aenv} \:\punc{.}
         \label{eq:3:thm:char:discr-preorder:jiparamsim:jiparamsimequiv} 
  \end{align}
\end{thm}
\begin{proof}
  We let 
    $\aenvLTS = \tuple{\envs,\actions,\senvltzero}$ be an environment \LTS,
      and we let $\aprocLTS = \tuple{\procs,\actions,\sprocltzero}$\vspace*{-2pt}
        be a process \LTS\ that contains the universal process $\univproc$.
      
  We first note that \eqref{eq:2:thm:char:discr-preorder:jiparamsim:jiparamsimequiv} 
    follows from, and is equivalent, to \eqref{eq:1:thm:char:discr-preorder:jiparamsim:jiparamsimequiv},
      because $\scanjiparamsim{\benv}$ is the converse relation of $\sjiparamsimby{\aenv}$.
  Furthermore \eqref{eq:3:thm:char:discr-preorder:jiparamsim:jiparamsimequiv} follows from 
              \eqref{eq:1:thm:char:discr-preorder:jiparamsim:jiparamsimequiv} and \eqref{eq:2:thm:char:discr-preorder:jiparamsim:jiparamsimequiv},
    due to $\jiparamsimequiv{\aenv} \,=\,  \jiparamsimby{\aenv} \cap \canjiparamsim{\aenv}$.
  Therefore it remains to prove \eqref{eq:2:thm:char:discr-preorder:jiparamsim:jiparamsimequiv}.
    For that we show the two directions of \ref{eq:1:thm:char:discr-preorder:jiparamsim:jiparamsimequiv}, and proceed as follows:
\begin{itemize}[label={``$\Rightarrow$'':},leftmargin=*,align=right,itemsep=0ex]
\item[``$\Rightarrow$'':]
  We show that
$\aenv \simby \benv \implies \sparamsimby{\benv} \:\subseteq\:\sparamsimby{\aenv}$. 
  Then \eqref{eq:1:thm:char:discr-preorder:jiparamsim:jiparamsimequiv} 
    follows from Proposition~\ref{prop:jiparamsim:equals:parasim}, \ref{it:1:prop:jiparamsim:equals:parasim}.
  %Our thesis then follows from item \ref{it:1:prop:jiparamsim:equals:parasim} of proposition \ref{prop:jiparamsim:equals:parasim}. 
So, let 
  %$r_{\aenv}$ 
  $\family{\sasimi{\aenv}}{\aenv\in\envs}$
  be a \parameterizedby{\aenvLTS} family of binary relations $\sasimi{\aenv} \subseteq \procs\times\procs$ that is defined, for all $\aenv\in\envs$ by:
\[
        \aproc \asimi{\aenv}\bproc %\quad\text{iff}\quad 
                                     \;\;\funin\Longleftrightarrow\;\;   
                                       \exists \benv \scansimulate \aenv \,
                                          \bigl[\aproc \paramsimby{\benv} \bproc\bigl] 
\]
It suffices to show that $\family{\sasimi{\aenv}}{\aenv\in\envs}$  %$r_{\aenv}$ 
                                    is an \parameterizedby{\aenvLTS} simulation.
So, suppose $\aproc \asimi{\aenv} \bproc$. Then $\aproc \paramsimby{\benv} \bproc$ for some $\benv \scansimulate \aenv$.
Suppose $\aproc \proclt{\aact} \aprocacc$ and $\aenv \envlt{\aact} \aenvacc$. Then $\benv \envlt{\aact} \benvacc$ for some 
$\benvacc \scansimulate \aenvacc$, and hence $\bproc \proclt{\aact}
\bprocacc$ for some $\bprocacc$ such that $\aprocacc \paramsimby{\benvacc} \bprocacc$. Then it follows that $\aprocacc \asimi{\aenvacc} \bprocacc$ holds, as required.

\item[``$\Leftarrow$'':]
  Assume $\sjiparamsimby{\benv} \:\subseteq\:\sjiparamsimby{\aenv}$.
    Let  $\aunivproc$ be a universal process in $\aprocLTS$.
      Then $\join{\aunivproc}{\benv} \bisim \benv$. 
      Moreover, $\benv \simby 
\join{\benv}{\benv}$ (which is easy to show), and then $\aunivproc \jiparamsimby{\benv} \benv$. By the assumption 
$\sjiparamsimby{\benv} \:\subseteq\:\sjiparamsimby{\aenv}$ we
have that $\aunivproc \jiparamsimby{\aenv} \benv$. In other words:
$\aenv \bisim \join{\aunivproc}{\aenv} \simby \join{\benv}{\aenv}$.
By Lemma~\ref{lem:aux1} we get $\aenv \simby \benv$, as required.\qed
\end{itemize}
\renewcommand{\qed}{}
\end{proof}

% \begin{proof}
%   \todo{TODO:} 
%     For the directions ``$\Longrightarrow$'' in \eqref{eq:thm:char:jiparamsim} and \eqref{eq:thm:char:jiparamsimequiv}
%     infinitary versions of a Hennessy--Milner-kind logical characterization of parameterized simulation $\sjiparamsimby{\aenv}$ \alert{needs to be checked again}
%     for the case of \alert{non-image-finite} environments.  
%     The directions ``$\Longleftarrow$'' in \eqref{eq:thm:char:jiparamsim} and \eqref{eq:thm:char:jiparamsimequiv}  
%       can be easily shown to hold also for the case of non-image-finite environments.
% \end{proof}

%% file: fig-ex1.tex
\begin{figure}[t!]
\begin{center}  
\begin{tikzpicture}

  \matrix[anchor=north,row sep=1cm,column sep=0.35cm,%every node/.style={draw,thick,circle,minimum width=2.5pt,fill,inner sep=0pt,outer sep=2pt},
            ampersand replacement=\& ] at (7.5,0) {
                   \&  \node(e-pos){}; \&                   \&[0.4cm] \&   \node(p-pos){};     \& \&[0.4cm]                \& \node(q-pos){}; \&               \& \&[0.5cm]     
                   \&[0.1cm] \node(pe-pos){}; \&[0.1cm]          \&[0.4cm]  \&                \& \&                \& \&[0.1cm] \node(qe-pos){};
    \\ 
    \node(e0-pos){};   \&              \&  \node(e1-pos){};     \&         \&   \node(p0-pos){};   \& \&        \node(q0-pos){};   \&             \&  \node(q1-pos){}; \& \&[0.5cm] 
    \node(pe0-pos){};  \&                     \&  \node(pe1-pos){};  \&[0.4cm]  \& \node(qe0-pos){};   \& \& \&  \node(qe1-pos){};  \&            \&[0.1cm] \node(qe2-pos){}; \& \& \& \node(qe3-pos){};              
    \\
    \node(e00-pos){};  \&              \&  \node(e10-pos){};    \&         \&   \node(p00-pos){};  \& \&        \node(q00-pos){};  \&             \&               \& \&[0.5cm]
    \node(pe00-pos){}; \&                     \&                 \&[0.4cm]  \& \node(qe00-pos){};
    \\  
  };
  %\draw[<-,very thick,>=latex,chocolate,shorten >=1pt](start) -- ++ (90:{0.5cm});
  %
  \path (e-pos) ++ (0cm,0cm) node[red](e){$\aenv$};
  \path (e0-pos) ++ (0cm,0cm) node[red](e0){$\bact$};
  \path (e00-pos) ++ (0cm,0cm) node(e00){$\deadlock$};
  \path (e1-pos) ++ (0cm,0cm) node(e1){$\deadlock$};
  \draw[-implies,double equal sign distance,firebrick](e) to node[left,pos=0.325]{$\aact$} (e0); 
  \draw[-implies,double equal sign distance](e) to node[right,pos=0.325]{$\aact$} (e1); 
  \draw[-implies,double equal sign distance](e0) to node[left,pos=0.4]{$\bact$} (e00); 
  \path (p-pos) ++ (0cm,0cm) node(p){$\aproc$};
  \path (p0-pos) ++ (0cm,0cm) node(p0){$\bact$};
  \path (p00-pos) ++ (0cm,0cm) node(p00){$\deadlock$};
  \draw[->,firebrick](p) to node[left,pos=0.4]{$\aact$} (p0); 
  \draw[->](p0) to node[left,pos=0.4]{$\bact$} (p00); 
  \path (q-pos) ++ (0cm,0cm) node(q){$\bproc$};
  \path (q0-pos) ++ (0cm,0cm) node(q0){$\bact$};
  \path (q00-pos) ++ (0cm,0cm) node(q00){$\deadlock$};
  \path (q1-pos) ++ (0cm,0cm) node(q1){$\deadlock$};
  \draw[->](q) to node[left,pos=0.325]{$\aact$} (q0); 
  \draw[->,firebrick](q) to node[right,pos=0.325]{$\aact$} (q1); 
  \draw[->](q0) to node[left,pos=0.4]{$\bact$} (q00);
  \draw[-,densely dashed,red,out=20,in=160,shorten >=0pt] 
    (p) to node{\scalebox{1.5}{$\times$}} node[pos=0.3,above]{$\aenv$} (q);
  \draw[-,densely dashed,red,out=35,in=150,shorten >=0pt] 
    (p0) to node[pos=0.4]{\scalebox{1.5}{$\times$}} node[pos=0.2,above]{$\bact$} (q1);

  \path(pe-pos) ++ (0cm,0cm) node(pe){$\join{\aproc}{\aenv}$};
  \path(pe0-pos) ++ (0cm,0cm) node(pe0){\small $\join{\bact}{\bact}$};
  \path(pe00-pos) ++ (0cm,0cm) node(pe00){\small $\join{\deadlock}{\deadlock}$};
  \path(pe1-pos) ++ (0cm,0cm) node(pe1){\small $ \join{\bact}{\deadlock}$};
  \draw[->](pe) to node[left,pos=0.325]{$\aact$} (pe0); 
  \draw[->](pe) to node[right,pos=0.325]{$\aact$} (pe1); 
  \draw[->](pe0) to node[left,pos=0.4]{$\bact$} (pe00); 
  \path(qe-pos) ++ (0cm,0cm) node(qe){$\join{\bproc}{\aenv}$};
  \path(qe0-pos) ++ (0cm,0cm) node(qe0){\small $\join{\bact}{\bact}$};
  \path(qe00-pos) ++ (0cm,0cm) node(qe00){\small $\join{\deadlock}{\deadlock}$};
  \path(qe1-pos) ++ (0cm,0cm) node(qe1){\small $ \join{\deadlock}{\bact}$};
  \path(qe2-pos) ++ (0cm,0cm) node(qe2){\small $ \join{\bact}{\deadlock}$};
  \path(qe3-pos) ++ (0cm,0cm) node(qe3){\small $ \join{\deadlock}{\deadlock}$};
  \draw[->](qe) to node[left,pos=0.325]{$\aact$} (qe0); 
  \draw[->](qe) to node[left,pos=0.5,xshift=2pt]{$\aact$} (qe1); 
  \draw[->](qe) to node[right,pos=0.5]{$\aact$} (qe2); 
  \draw[->](qe) to node[right,pos=0.325,xshift=1.5pt]{$\aact$} (qe3); 
  \draw[->,shorten >=-2pt](qe0) to node[right,pos=0.2,xshift=-2pt]{$\bact$} (qe00);   
  \draw[-,thick,densely dashed,magenta,out=20,in=160,shorten >=-2pt] 
    (pe) to (qe);
  \draw[-,thick,densely dashed,magenta,out=30,in=150,shorten >=-2pt] 
    (pe0) to (qe0);
  \draw[-,thick,densely dashed,magenta,out=-40,in=225,shorten >=-2pt] 
    (pe1) to (qe1); 
  \draw[-,thick,densely dashed,magenta,out=-40,in=200,shorten >=-2pt,] 
    (pe1) to (qe2); 
  \draw[-,thick,densely dashed,magenta,out=-40,in=195,shorten >=-2pt] 
    (pe1) to (qe3); 
  \draw[-,thick,densely dashed,magenta,out=-10,in=190,shorten >=-2pt] 
    (pe00) to (qe00);
\end{tikzpicture}  
\end{center}  
  \vspace*{-3ex}
  \caption{\label{fig:lem:jiparamsim:subseteq:paramsim}%
           Example for witnessing $\sparambisim{\aenv} \;\neq\; \sjiparambisim{\aenv}\,$:
           For $\aenv \defdby \actpref{\aact}{\bact} \plus \aact$,
               $\aproc \defdby \actpref{\aact}{\bact}$,
           and $\bproc \defdby \aenv$
           it holds that $\aproc \notparambisim{\aenv} \bproc$
             (indicated by the \alert{mismatches} ${\bf \alert{{\times}}}$ when building a parameterized bisimulation),
             but $\join{\aproc}{\aenv} \bisim \join{\bproc}{\aenv}$
               (indicated by the \magenta{bisimulation links})
             and hence $\aproc \jiparambisim{\aenv} \bproc$.}%
             \label{fig:1}
\end{figure}

%% file: log-char.tex
In this section we adapt the modal characterization of parameterized bisimilarity~$\sparambisim{\aenv}$,
  see Theorem~\ref{thm:log-char:parambisim}, 
   %due to Stirling and developed by Larsen in \cite{lars:1986,lars:1987}, see Theorem~\ref{thm:log-char:parambisim}, 
    % that Larsen developed in \cite{lars:1986,lars:1987} and credits to Colin Stirling, see Theorem~\ref{thm:log-char:parambisim},  
    for \jiorparameterized\ simulatability~$\sparamsimby{\aenv}$ and $\sjiparamsimby{\aenv}$. 
The crucial observation for our adaptation is Lemma~\ref{lem:posforms:of:joins} below
  which states that the set of positive formulas that a join interaction $\join{\aproci{1}}{\aproci{2}}$ satisfies
    is the intersection of the sets of positive formulas satisfied by the constituent processes $\aproci{1}$ and $\aproci{2}$. 
Finally we explain why a similar line of argument is not possible
  in order to adapt the modal characterization for $\sparambisim{\aenv}$ to obtain one for \jiparameterized\ bisimilarity~$\sjiparambisim{\aenv}$.  
In doing so we provide some evidence for Larsen's assessment,
  that (in view of that $\sparambisim{\aenv} \,\subsetneqq\, \sjiparambisim{\aenv}$ holds in general, see Proposition~\ref{prop:parambisim:jiparambisim:jiparamsimequiv})
    ``the modal characterization for $\sparambisim{\aenv}$ does not hold for $\sjiparambisim{\aenv}$, 
      and no other modal characterization seems immediate'' \cite[p.$\,$210]{lars:1987}.
However, we report about further work of ours on this issue in the final section (see \ref{CW1} in Section~\ref{concl}).

% \begin{lem}
%   The set of positive formulas that are valid for a \joininteraction\ $\join{\aproci{1}}{\aproci{2}}$ 
%       is the intersection of the positive formulas valid for the constituent processes $\aproci{1}$ and for $\aproci{2}$.
%   More formally, for all processes $\aproci{1}$ and $\aproci{2}$ it holds:
%   %
%   \begin{equation*}
%     \posformsof{ \join{\aproci{1}}{\aproci{2}} }
%       =
%     \posformsof{ \aproci{1} }   \cap   \posformsof{ \aproci{2} } \punc{.}  
%   \end{equation*}
%   %
% \end{lem}

\begin{lem}\label{lem:posforms:of:joins}
  For all processes $\aproc$ and $\bproc$ of a process LTS $\aprocLTS = \tuple{\procs,\actions,\sprocltzero}$ it holds:
  \begin{equation}\label{eq:lem:posforms:of:joins}
    \posformsof{ \join{\aproc}{\bproc} }
      =
    \posformsof{ \aproc }   \cap   \posformsof{ \bproc } \punc{.}  
  \end{equation}
  %
  % That is, 
  %   the set of positive formulas valid for join interactions $\join{\aproci{1}}{\aproci{2}}$ 
  %     is the intersection of the positive formulas valid for $\aproci{1}$ and for $\aproci{2}$.
  %     %
\end{lem}

\begin{proof}
  Statement \eqref{eq:lem:posforms:of:joins} can be established by induction on the structure of positive modal formulas~$\aform$
    according to their definition in grammar \eqref{eq:2:def:forms} of Definition~\ref{def:forms}.
      
  The base case of \eqref{lem:posforms:of:joins} for $\aform = \True$ is obviously true, because  $\True$ is satisfied for any process.
  It remains to establish the induction step for formulas of the forms $\aform = \aformi{1} \logand \aformi{2}$ and $\aform = \diamondact{\aact} \aformi{0}$.
  Since in the first case the induction step is easy to demonstrate,
    we only treat the more interesting case of  $\aform = \diamondact{\aact} \aformi{0}$.
  For this we argue as follows:
  \begin{align*}
    &
    \diamondact{\aact} \aformi{0} \in \posformsof{ \join{\aproc}{\bproc} }
      \;\; \Longleftrightarrow \;\;
        \join{\aproc}{\bproc} \satisfies \diamondact{\aact} \aformi{0}
    \\
    & \hspace*{1.25em}  
      \;\; \Longleftrightarrow \;\;
        (\existsstzero{\,\aprocacc,\bprocacc\in\procs}) \, 
          \bigl[\, \aproc \lt{\aact} \aprocacc
                     \,\logand\,
                   \bproc \lt{\aact} \bprocacc
                     \,\logand\,
                   \join{\aprocacc}{\bprocacc} \satisfies \aformi{0}
          \,\bigr]
      \displaybreak[0]\\
    & \hspace*{1.25em}  
      \;\; \Longleftrightarrow \;\;
        (\existsstzero{\,\aprocacc,\bprocacc\in\procs}) \, 
          \bigr[\, \aproc \lt{\aact} \aprocacc
                     \,\logand\,
                   \bproc \lt{\aact} \bprocacc
                     \,\logand\,
                   \aformi{0} \in \posformsof{ \join{\aprocacc}{\bprocacc} }
          \,\bigr]
      \displaybreak[0]\\[-0.25ex]
    & \hspace*{1.25em}  
      \;\; \overset{\text{IH}}{\Longleftrightarrow} \;\;
        (\existsstzero{\,\aprocacc,\bprocacc\in\procs}) \, 
          \bigl[\, \aproc \lt{\aact} \aprocacc
                     \,\logand\,
                   \bproc \lt{\aact} \bprocacc
                     \,\logand\,
                   \aformi{0} \in \posformsof{\aprocacc} \cap \posformsof{\bprocacc}
           \,\bigr]
      \displaybreak[0]\\
    & \hspace*{1.25em}  
      \;\; \Longleftrightarrow \;\;
        (\existsstzero{\aprocacc\in\procs}) \, 
          \bigl[\,
             \aproc \lt{\aact} \aprocacc
               \,\logand\,
             \aformi{0} \in \posformsof{\aprocacc}
          \,\bigr] \logand
        (\existsstzero{\bprocacc\in\procs}) \,     
          \bigr[\, 
             \bproc \lt{\aact} \bprocacc
               \,\logand\,
             \aformi{0} \in \posformsof{\bprocacc}
          \,\bigr]
      \displaybreak[0]\\
    & \hspace*{1.25em}  
      \;\; \Longleftrightarrow \;\;
        \diamondact{\aact} \aformi{0} \in \posformsof{\aproc}
          \;\logand
        \diamondact{\aact} \aformi{0} \in \posformsof{\bproc}  
    \\
    & \hspace*{1.25em}  
      \;\; \Longleftrightarrow \;\;
        \diamondact{\aact} \aformi{0} \in \posformsof{\aproc} \cap \posformsof{\bproc} \punc{,}   
  \end{align*}
  where we have marked by (IH) the logical equivalence in which the induction hypothesis is used.
\end{proof}

Based on this lemma, a modal characterization of $\sjiparamsimby{\aenv}$ and $\sparamsimby{\aenv}$ is now an easy consequence
  of the modal characterization of the simulation preorder $\ssimby$ on processes, see \eqref{eq:sim:thm:log-char:sim:bisim} in Theorem~\ref{thm:log-char:sim:bisim}.
In this way we obtain,
  in analogy with Theorem~$\,$\ref{thm:log-char:parambisim} ,
    the following modal characterizations
      of \jiorparameterized\ similarity and of \jiorparameterized\ simulation equivalence
        with respect to positive formulas.
      
\begin{thm}\label{thm:log-char:paramsim:paramsimequiv}%
  For all \imagefinite\ environments~$\aenv$, 
    the following characterizations of $\sjiparamsimby{\aenv}$ and $\sparamsimby{\aenv}\,$,
      as well as of $\sjiparamsimequiv{\aenv}$ and $\sparamsimequiv{\aenv}$ hold
        for all \imagefinite\ processes $\aproc$, $\bproc\,$:
  \begin{alignat}{5}
    \aproc & \jiparamsimby{\aenv} \bproc \;\;\;
    &  
    \bigl( & \Longleftrightarrow\;\; 
        &
        \aproc & \paramsimby{\aenv} \bproc 
    \,\bigr) &
    \;\; & \Longleftrightarrow\;\; & 
    \posformsof{\aproc}
      \cap
    \posformsof{\aenv}
      \: & \subseteq \:
    \posformsof{\bproc}
      \cap
    \posformsof{\aenv} \punc{,}
      \label{eq:1:thm:log-char:paramsim:paramsimequiv}
    \displaybreak[0]\\
    \aproc & \jiparamsimequiv{\aenv} \bproc \;\;\;
    &  
    \bigl( & \Longleftrightarrow\;\; 
        &
        \aproc & \paramsimequiv{\aenv} \bproc 
    \,\bigr) &
    \;\; & \Longleftrightarrow\;\; & \;\;
    \posformsof{\aproc}
      \cap
    \posformsof{\aenv}
      \: & = \:
    \posformsof{\bproc}
      \cap
    \posformsof{\aenv} \punc{.}
      \label{eq:2:thm:log-char:paramsim:paramsimequiv}
  \end{alignat}      
  The implications ``$\Rightarrow$'' in \eqref{eq:1:thm:log-char:paramsim:paramsimequiv} and \eqref{eq:2:thm:log-char:paramsim:paramsimequiv} 
    hold also for not necessarily \imagefinite\ $\aproc$, $\bproc$, and $\aenv$.
    %
  % Due to $\sjiparamsimby{\aenv} = \sparamsimby{\aenv}$ (see Lem.~\ref{lem:jiparamsim:vs:paramsim}, \ref{it:2:lem:jiparamsim:vs:paramsim})
  %   the statement \eqref{eq:1:thm:log-char:paramsim} 
  %     also characterizes $\aproc \jiparamsimby{\aenv} \bproc$,
  % and consequently 
  % $\sjiparamsimequiv{\aenv} = \sparamsimequiv{\aenv}$ (see Lem.~\ref{lem:jiparamsim:vs:paramsim}, \ref{it:3:lem:jiparamsim:vs:paramsim})   
  %   the statements \eqref{eq:2:thm:log-char:paramsim} 
  %     also characterizes $\aproc \jiparamsimequiv{\aenv} \bproc$.
\end{thm}

\begin{proof}
  For \eqref{eq:1:thm:log-char:paramsim:paramsimequiv}
    we argue as follows for all \imagefinite\ processes $\aproc$ and $\bproc$, and environments~$\aenv\,$:
  \begin{alignat*}{2}
    \aproc \paramsimby{\aenv} \bproc
      \;\; & \Longleftrightarrow \;\;
    \aproc \jiparamsimby{\aenv} \bproc
      & & \text{(by Prop.~\ref{prop:jiparamsim:equals:parasim}, \ref{it:1:prop:jiparamsim:equals:parasim})}
    \displaybreak[0]\\
      \;\; & \Longleftrightarrow \;\;
    \join{\aproc}{\aenv} \simby \join{\bproc}{\aenv}
      & & \text{(by the definition of $\sjiparamsimby{\aenv}$)}
    \displaybreak[0]\\
      \;\; & \Longleftrightarrow \;\;
    \posformsof{ \join{\aproc}{\aenv} } \subseteq \posformsof{ \join{\bproc}{\aenv} }
      & & \text{(by \eqref{eq:sim:thm:log-char:sim:bisim} in Thm.~\ref{thm:log-char:sim:bisim})}
    \displaybreak[0]\\
      \;\; & \Longleftrightarrow \;\;
    \posformsof{\aproc} \cap \posformsof{\aenv}
      \subseteq
    \posformsof{\bproc} \cap \posformsof{\aenv}
      \qquad
      & & \text{(by using Lem.~\ref{lem:posforms:of:joins}).}
  \end{alignat*}
  The implication ``$\Rightarrow$'' in the third equivalence statement holds also for not necessarily \imagefinite\ $\aproc$, $\bproc$, and $\aenv$
    due to the Hennessy--Milner Theorem~\ref{thm:log-char:sim:bisim}.
  Together with the fact that ``$\Rightarrow$'' also holds for the other three equivalence statements above,
    this demonstrates that ``$\Rightarrow$'' in \eqref{eq:1:thm:log-char:paramsim:paramsimequiv} holds for all $\aproc$, $\bproc$, and $\aenv$.
  
  Statement \eqref{eq:2:thm:log-char:paramsim:paramsimequiv} for the \jiorparameterized\ simulation equivalences $\sjiparamsimequiv{\aenv}$ and $\sparamsimequiv{\aenv}$
    follows from \eqref{eq:1:thm:log-char:paramsim:paramsimequiv}
  due to the definition of $\sjiparamsimequiv{\aenv}$ from $\sjiparamsimby{\aenv}$ in \eqref{eq:jiparamsimequiv:def:jiparamsimby:jiparamsimequiv:jiparambisimby},
                    and of $\sparamsimequiv{\aenv}$ from $\sparamsimby{\aenv}$ in Definition~\ref{def:paramsim}. 
\end{proof}

There is no obvious generalization of Lemma~\ref{lem:posforms:of:joins}
  that applies to all formulas of $\forms$. In particular,
    $\formsof{ \join{\aproci{1}}{\aproci{2}} }
       = \formsof{\aproci{1}}
           \cap
         \formsof{\aproci{2}}$
    does not hold, because certainly ``$\subseteq$'' is violated:
      in case that $\aproci{1}$ and $\aproci{2}$ are such that $\aproci{1} \permitslt{\aact}$ and $\aproci{2} \notpermitslt{\aact}$ holds for some $\aact\in\actions$,
        then $\lognot{\diamondact{\aact} \True} \in \formsof{ \join{\aproci{1}}{\aproci{2}} }$ due to $(\join{\aproci{1}}{\aproci{2}}) \notpermitslt{\aact}$,
        but $\lognot{\diamondact{\aact} \True} \notin \formsof{\aproci{1}}$, and hence
            $\lognot{\diamondact{\aact} \True} \notin \formsof{\aproci{1}} \cap \formsof{\aproci{2}}$.
Therefore the proof of the characterization above cannot 
  be extended, at least not in an analogous manner, to obtain a modal characterization of \jiparameterized\ bisimilarity~$\sjiparambisim{\aenv}$. 
(But see the report about our current work \ref{CW1} in Section~\ref{concl}.) 

Yet an interesting specialization of Lemma~\ref{lem:posforms:of:joins}
  concerns the specialized join operation $\sjoindot$ introduced in Definition~\ref{def:joindot}:
    $\actprojform{\formsof{ \joindot{\aproc}{\aenv} }}
       = \formsof{\aproc}
           \cap
         \negclosure{\posformsof{\aenv}}$
    holds for all processes~$\aproc$ and environments~$\aenv$,
      where $\actprojform{\cdot}$ projects modalities $\diamondact{\pair{\aact}{\aenv}}$ in formulas
        to their action components $\diamondact{\aact}$. 
  This observation can be used, together with the characterization of parameterized bisimilarity~$\sparambisim{\aenv}$ 
    via $\sjoindot$ in \eqref{eq:1:lem:joindot:join} of Lemma~\ref{lem:joindot:join},
      to obtain, for Larsen's characterization of $\sparambisim{\aenv}$ in Theorem~\ref{thm:log-char:parambisim}, an alternative proof 
        that is similar to the proof of Theorem~\ref{thm:log-char:paramsim:paramsimequiv} above.

%% file: concl.tex
Here we first summarize our contributions in a list with references to statements in earlier sections.
We then explain our path to the definition of \jiparameterized\ bisimilarity,
  and Larsen's comments on the shortcomings of this concept.
Furthermore we collect some references to the literature 
  concerning work that has been done based on parameterized bisimilarity in the meantime. 
Subsequently we report about our current work on a modal characterization of \jiorparameterized\ bisimilarity,
  and about generalizations of the modal characterizations here and by Stirling and Larsen.
We also describe some open problems of which the solutions have evaded us thus far.
Finally we mention our plan to investigate whether \jiparameterized\ bisimilarity
  can be used to refine Larsen's results in his thesis~\cite{lars:1986}
    on a method to show program correctness under the formation of contexts.

\smallskip
\noindent
{\bf Contribution.}
  Below we provide a summary by listing the concepts that we have defined and the results we have obtained,
    together with references to the appertaining formal statements:
\begin{enumerate}[label={(C\arabic{*})},leftmargin=*,align=right,itemsep=0ex]%[label={(\alph{*})},itemsep=0ex]
  \item{}\label{C1}
    We complemented Larsen's parameterized bisimilarity~$\sparambisim{\aenv}$ with respect to `synchronous' interaction with environments~$\aenv$
      by also defining parameterized simulatability, 
        the simulation \preorder~$\sparamsimby{\aenv}$,
          on processes with respect to `synchronous' interaction with environment~$\aenv$
            (see Definition~\ref{def:paramsim}).
  \item{}\label{C2}
    We defined weaker versions $\sjiparamsimby{\aenv}$ of $\sparamsimby{\aenv}$ 
                           and $\sjiparambisim{\aenv}$ of $\sparambisim{\aenv}$
        by relaxing the synchronicity condition of environment interaction for $\sparamsimby{\aenv}$ and $\sparambisim{\aenv}$
          to require only the existence of simulations, and respectively of bisimulations,
            between free join interactions ($\sjoin$) with environments~$\aenv$
              (see Definition~\ref{def:jiparamsimby:jiparamsimequiv:jiparambisimby}).
  \item{}\label{C3}
    We showed that $\sparamsimby{\aenv}$ and $\sparambisim{\aenv}$
      can be characterized similarly to the definitions of $\sjiparamsimby{\aenv}$, and $\sjiparambisim{\aenv}$ via join interactions ($\sjoin$) 
        as the existence of a simulation, and as bisimilarity, respectively,
          of free interactions with the specific form $\sjoindot$ (see Definition~\ref{def:joindot}) of join interactions
            that record targets of environment transitions in action labels (see Lemma~\ref{lem:joindot:join}).
              %
    % We show that the synchronicity of environment interactions that underlies the definition of parameterized bisimilarity~$\sparambisim{\aenv}$
    %   corresponds precisely to bisimilarity of free interactions with a specific form $\sjoindot$ of join interactions
    %     that records targets of environment transitions in action labels (see Lemma~\ref{lem:parambisim:jiparambisim}).
   
  \item{}\label{C4}
    We established that $\sparambisim{\aenv}$ and $\sjiparambisim{\aenv}$ coincide for deterministic environments~$\aenv$
      (see Proposition~\ref{prop:parambisim:jiparambisim:det:envs}). 
    
  \item{}\label{C5}
    We settled the relationships between \jiorparameterized\ bisimilarity $\sparambisim{\aenv}$ and $\sjiparambisim{\aenv}$,
      and the \jiorparameterized\ simulation equivalences $\sparamsimequiv{\aenv}$ and $\sjiparamsimequiv{\aenv}\,$:
      for all environments $\sparambisim{\aenv}$ is contained in $\sjiparambisim{\aenv}$, and furthermore
                           $\sjiparambisim{\aenv}$ is contained in both of $\sparamsimequiv{\aenv}$ and $\sjiparamsimequiv{\aenv}\,$, which coincide.
      The two inclusions in this chain are proper in general. (See Theorem~\ref{thm:incl:jiorparam:bisim:sim}).   
      %
      % Parameterized similarity~$\sparambisim{\aenv}$ is contained in \jiparameterized\ bisimilarity~$\sjiparambisim{\aenv}$.
      %   Then $\sjiparambisim{\aenv}$ is contained in $\sparamsimequiv{\aenv}$,
      %     which however coincides with \jiparameterized\ simulation equivalence $\sjiparamsimequiv{\aenv}$.
      % The two inclusions in this chain are proper in general. (See Theorem~\ref{thm:incl:jiorparam:bisim:sim}).                           
      %
  \item{}\label{C6}    
    Larsen's main technical result about $\sparambisim{\aenv}$ (see Theorem~\ref{thm:char:discr-preorder:parambisim}),
      % which states 
      that the discrimination \preorder\ induced by $\sparambisim{\aenv}$ on environments coincides with the simulation preorder~$\ssimby$ on environments,
        does not hold analogously for the discrimination \preorder\ induced by $\sjiparambisim{\aenv}$ 
          (see Proposition~\ref{prop:discr-preorder:jiparambisim}). 
    However, we showed that this coincidence with the simulation \preorder~$\ssimby$ on environments
      \emph{does} hold analogously for the discrimination \preorders\ induced both by \jiorparameterized\ similarity $\sparamsimby{\aenv} = \sjiparamsimby{\aenv}$
                                                                      and by \jiorparameterized\ simulation equivalence $\sparamsimequiv{\aenv} = \sjiparamsimequiv{\aenv}$
                                                                      (see Theorem~\ref{thm:char:discr-preorder:jiparamsim:jiparamsimequiv}). 
  \item{}\label{C7}
    We adapted Stirling and Larsen's modal characterization of parameterized bisimilarity~$\sparambisim{\aenv}$ (see Theorem~\ref{thm:log-char:parambisim})
      to obtain a modal characterization of \jiorparameterized\ similarity $\sparamsimby{\aenv} = \sjiparamsimby{\aenv}$
        and also of \jiorparameterized\ simulation equivalence $\sparamsimequiv{\aenv} = \sjiparamsimequiv{\aenv}$
        (see Theorem~\ref{thm:log-char:paramsim:paramsimequiv}).
\end{enumerate}

% Due to the modal characterization presented in the previous section and the simple characterization of the discrimination ordering presented in this section, 
% we are confident that the notion of parameterized bisimulation equivalence proposed is indeed a natural one. To give further support for the proposed parameterized version of bisimulation equivalence, let us consider an alternative and perhaps more immediate parameterized version.

\smallskip
\noindent
{\bf Larsen on \jiparameterized\ bisimilarity, and our way to its definition.}
  We formulated \jiparameterized\ bisimilarity and \jiparameterized\ simulatability
    while reading Larsen's article \cite{lars:1987} from 1987, and trying to improve our intuitive understanding of parameterized bisimilarity.
  Afterwards we developed, in stages, the results that we report here.
  Only when diving deeper into the intricate %(and very subtle) 
                                             proof of Larsens main result, 
    the characterization of the discrimination \preorder\ induced by parameterized bisimilarity~$\sparambisim{\aenv}$
      as simulatability of environments (see Theorem~\ref{thm:char:discr-preorder:parambisim}),
        did we find his remarks about an ``alternative and perhaps more immediate parameterized version [of bisimulation equivalence]''.
  This passage appears on page~210 in \cite{lars:1987}, at the end of Section~5 that is devoted to this central result.
  The version of bisimulation equivalence that Larsen sketches there coincides with \jiparameterized~bisimilarity~$\sjiparambisim{\aenv}$.
    
  Larsen refers to \jiparameterized~bisimilarity 
    in order to ``give further support for the proposed parameterized version of bisimulation equivalence'',
      in addition to the following assessment:
        ``Due to the modal characterization presented [\ldots] and the simple characterization of the discrimination ordering presented [\ldots], 
         we are confident that the notion of parameterized bisimulation equivalence proposed is indeed a natural one.''
  As for the mentioned further evidence Larsen notes that \jiparameterized\ bisimilarity ``lacks many of the properties presented in this paper''.
  Concretely he mentions three properties.
  
  First, that ``$\sparambisim{\aenv}$ is strictly included in $\sjiparambisim{\aenv}$ for all environments~$\aenv$''
    (in general is meant [we use our notation for $\sjiparambisim{\aenv}$ here]), 
     corresponding to Proposition~\ref{prop:parambisim:jiparambisim:jiparamsimequiv}, 
                                  \ref{it:1:prop:parambisim:jiparambisim:jiparamsimequiv} and \ref{it:2:prop:parambisim:jiparambisim:jiparamsimequiv}). 
  Second,
    that ``thus the modal characterization for $\sparambisim{\aenv}$ does not hold for $\sjiparambisim{\aenv}\,$,
          and no other modal characterization seems immediate.''
    This assessment stimulates us to work out \ref{CW1}.
   
  Finally third, Larsen writes that:
    ``More important though is that the simulation ordering does not characterize the discrimination ordering generated by this alternative parameterized version[.]'',
    in contrast with his impressive and surprising main result in \cite{lars:1987}, Theorem~\ref{thm:char:discr-preorder:parambisim} here,  
      which shows that that is the case for parameterized bisimilarity~$\sparambisim{\aenv}$. 
    For this observation Larsen uses a counterexample
      that is slightly different from the one we use for Proposition~\ref{prop:discr-preorder:jiparambisim}, the corresponding statement here.
Ji-parameterized bisimilarity fits nicely in the recently proposed 
framework \cite{abs-2007-08187,AubertV22}: the authors 
thereof advocate to make a clear distinction between processes and 
what they call tests. The composition of processes with tests give 
rise to instrumentations, which are to be understood as ``compiled binaries ready to be executed'' \cite[page 6]{AubertV22}. Environments in our setting
can be seen as tests, and hence join is a notion of composition. In
contrast, $\sparambisim{\aenv}$ lacks the composition operation (but still distinguishes processes and tests/environments).
    Below we formulate the question of a characterization of the discrimination \preorder\ induced by \jiparameterized\ bisimilarity
      as the open problem \ref{P1}, and a specialization of this question as the open problem~\ref{P2}.
  
%  We previously used the communication merge symbol ``$\scommmerge$'' 
%    for defining `free-en\-vi\-ron\-ment-in\-ter\-ac\-tion parameterized bisimilarity' $\sfeiparambisim{\aenv}$
%      via $\aproc \feiparambisim{\aenv} \bproc \,\Leftrightarrow\, \aproc \commmerge \aenv \bisim \bproc \commmerge \aenv$.
%  After finding the cited passage in \cite{lars:1987}, 
%We adopted Larsen's notation in a slightly adapted form,
%    and in particular took over the terminology of the join operation ``$\sjoin$'' from his thesis \cite{lars:1986} from 1986
%      (called `conjunction' in the 1987-article \cite{lars:1987}).
%  

\smallskip
\noindent
{\bf Literature on parameterized versions of bisimilarity.}
  Parameterized bisimilarity proved to be a very fruitful concept since its inception by Larsen in \cite{lars:1986,lars:1987}.
  His definition has been applied, specialized, and adapted in multiple ways in the meantime. Please see below for a few examples.
  But to the best of our knowledge this does not hold for the \jiparameterized\ concepts of simulatability and bisimilarity,
        apart from the passages in \cite{lars:1987} that we cited and described above.
       
Parameterized bisimilarity in Larsen's definition \cite{lars:1986,lars:1987}
   has later been called `relative bisimilarity' and `relativized bisimilarity' in \cite{lars:miln:1992} by Larsen and Milner,
     who used it also for the practical purpose of verifying the Alternating Bit Protocol \cite{lars:miln:1992}. 
As pointed out in \cite{fahr:lars:lega:trao:2014},
  it was also the basis for `modal transition systems' to which a large body of work has been devoted since, 
    see for example \cite{lars:1990,lars:nyma:wkas:2007,bene:kvre:lars:moll:srba:2011,huth:jaga:schm:2001}.

Environment parameterized bisimulations in the sense of Larsens' definition or adapted and specialized variants of it have been used frequently, 
  for example in \cite{lars:miln:1992,lars:lars:waso:2005}. 
  \cite{pier:sang:2000} introduces a notion of equivalence parameterized with respect to typing information, 
    which, quoting from \cite{pier:sang:2000}: 
      ``can be seen as a disciplined instance of Larsen’s, in which one uses types to express constraints on the behaviors of the observers, 
        rather than explicitly writing all their possible behavior''.

\smallskip
\noindent
{\bf Current Work.}
  We investigate \modallogical\ characterizations
  of \jiorparameterized\ bisimilarity and simulatability,
    and of refinements of the modal characterizations already obtained by Larsen, and here.
\begin{enumerate}[label={(W\arabic{*})},leftmargin=*,align=right,itemsep=0ex]  
  \item{}\label{CW1}
    % About logical characterizations, Larsen writes in \cite{lars:1987}:
    %    ``Intuitively, for $\aproc \parambisim{\aenv} \bproc$ to hold, 
    %        $\aenv$ must interact identically with $\aproc$ and $\bproc$, 
    %          whereas this is not required for $\aproc \jiparambisim{\aenv} \bproc$ to hold.
    %      Thus, the modal characterization for $\sparambisim{\aenv}$ does not hold for $\sjiparambisim{\aenv}$, 
    %        and no other modal characterization seems~immediate.'' (We have adapted the notation to the one we introduced here.)
    %       
    We are working out a \modallogical\ characterization for \jiparameterized\ bisimilarity $\sjiparambisim{\aenv}$ %for environments $\aenv$
      that is based on a game characterization of $\sjiparambisim{\aenv}$.
    However, our characterization will \underline{not} just be of a simple form comparable
      to Theorem~\ref{thm:log-char:parambisim} and Theorem~\ref{thm:log-char:paramsim:paramsimequiv},
        for all environments~$\aenv\,$:
    \begin{equation*}
      \aproc \jiparambisim{\aenv} \bproc 
        \;\;\Longleftrightarrow\;\; 
      \formsof{\aproc}
        \cap
      \someformsof{\aenv}
        \: = \:
      \formsof{\bproc}
        \cap
      \someformsof{\aenv} 
        \quad
        \text{(for all \imagefinite\ processes $\aproc$, $\bproc$).}
          %\label{eq:log-char:jiparambisim::easy:form}
    \end{equation*}
    where $\someforms \subseteq \forms$ would be appropriately defined formulas
      with then $\someformsof{\cenv} \defdby \descsetexp{ \aform\in\someforms }{ \cenv \satisfies \aform }$
        defined for all environments~$\cenv$.
    % It seems that a characterization of $\sjiparambisim{\aenv}$ of the form \eqref{eq:log-char:jiparambisim::easy:form} is impossible.
    It is nevertheless interesting to note that
      since $\snotjiparambisim{\aenv} \subseteq \snotparambisim{\aenv}$ holds 
        (due to $\sparambisim{\aenv} \subseteq \sjiparambisim{\aenv}$ by Proposition~\ref{prop:parambisim:jiparambisim:jiparamsimequiv}, \ref{it:1:prop:parambisim:jiparambisim:jiparamsimequiv}),
      that whenever $\aproc \notjiparambisim{\aenv} \bproc$ holds,
             always $\aproc \notparambisim{\aenv} \bproc$ follows,
        and a formula $\aform \in (\formsof{\aproc} \cap \negclosure{\posformsof{\aenv}}) \mathrel{\Delta} (\formsof{\bproc} \cap \negclosure{\posformsof{\aenv}})$
          (where $\Delta$ denotes symmetric difference)
          that distinguishes $\aproc$ and $\bproc$ can always be found via Larsen's characterization, Theorem~\ref{thm:log-char:parambisim}.      
  \item{}\label{CW2}
    The restriction to \imagefinite\ processes 
      for the modal characterizations of simulatability~$\ssimby$ and bisimilarity~$\sbisim$
        by Hennessy and Milner (Theorem~\ref{thm:log-char:sim:bisim})
          can be dropped 
            by permitting infinitary formulas with infinite conjunctions.
              Indeed, Milner has described such an adaptation for infinitary formulas~in~\cite{miln:1985}.
              
    We want to obtain similar extensions to not necessarily \imagefinite\ processes 
      for Larsen's characterization of $\sparambisim{\aenv}$ (Theorem~\ref{thm:log-char:parambisim})
        and our ones of $\sparamsimby{\aenv} \,=\, \sjiparamsimby{\aenv}$ 
          and $\sparamsimequiv{\aenv} \,=\, \sjiparamsimequiv{\aenv}$ (Theorem~\ref{thm:log-char:paramsim:paramsimequiv}).
\end{enumerate}

\smallskip
\noindent
{\bf Open problems.}        
As problems to which (satisfactory) answers have evaded us so far, we want to mention:
\begin{enumerate}[label={(P\arabic{*})},leftmargin=*,align=right,itemsep=0ex]       
  \item{}\label{P1}
    How can the discrimination order for $\sjiparambisim{\aenv}$ be characterized?
    Note that a similar characterization in terms of simulatability~$\ssimby$
      as for the discrimination order of $\sparambisim{\aenv}$ in Thm.~\ref{thm:char:discr-preorder:parambisim} by Larsen,
        and for $\sjiparamsimby{\aenv}$ and $\sjiparamsimequiv{\aenv}$ in Theorem.~\ref{thm:char:discr-preorder:jiparamsim:jiparamsimequiv},
          is not possible due to Proposition~\ref{prop:discr-preorder:jiparambisim}.
          %  
          % We have so far not succeeded in obtaining a characterization of the discrimination order for $\sjiparambisim{\aenv}$
          %   similar to the characterizations of the discrimination order for $\sparambisim{\aenv}$ in Thm.~\ref{thm:char:discr-preorder:parambisim} by Larsen,
          %     and for $\sparamsimby{\aenv}$ and $\sparamsimequiv{\aenv}$ in Theorem.~\ref{thm:char:discr-preorder:jiparamsim:jiparamsimequiv}.
          %        
  \item{}\label{P2}
    Does equality of \jiparameterized\ bisimilarity with respect to environments~$\aenv$ and~$\benv$
      coincide with bisimilarity of $\aenv$ and~$\benv\,$?
        Equivalently, does the implication ``$\Leftarrow$'' hold in the following statement
          (of which ``$\Rightarrow$'' is easy to verify), for all environments~$\aenv$ and $\benv\,$:
    \begin{equation*}
      \aenv \bisim \benv
        \;\;\; \overset{\displaystyle \overset{?}{\Longleftarrow}}
                       {\:\:\Longrightarrow} \;\;\;
      \jiparambisim{\aenv} \;\, = \;\, \jiparambisim{\benv} \punc{.}                 
    \end{equation*}   
    %    
    %   Neither are we sure about the implication ``$\Leftarrow$'' in the following statement:
    %   %
    %   \begin{equation}
    %     \aenv \bisim \benv
    %       \;\;\:\overset{?}{\Longleftrightarrow}\;\;\;\:
    %          %(\,%
    %          \sjiparambisim{\aenv}
    %            \;\, = \;\,
    %          \sjiparambisim{\benv} 
    %          %\,) 
    %          \punc{,}
    % \end{equation}
    % %
    % in which ``$\Rightarrow$'' holds, and is easy to show.
  %
\end{enumerate}

\smallskip
\noindent
{\bf Future research.}       
  As two lines of research for which the concept of \jiparameterized\ bisimilarity may lead to new insights 
    we mention:
      a continuation of Larsen's work in his thesis \cite{lars:1986} 
        towards flexible formal methods for showing compositionality of program correctness (see (F1)),
      and consequences for finding interesting contextual behavioural metrics as introduced in \cite{lago:murg:2023} (see (F2)):  
\begin{enumerate}[label={(F\arabic{*})},leftmargin=*,align=right,itemsep=0ex]      
  \item{}\label{F1}
%    \todo{context}
    An interesting future work is the study of compositionality 
    properties of $\sjiparambisim{\aenv}$, that is the behavior of 
    $\sjiparambisim{\aenv}$ up to context. A context is typically
    defined as a syntactic process $C$ (expressed in some process 
    algebra) with a hole $[]$. Notation $C[\aproc]$ is used 
    for the process obtained upon substitution of $\aproc$ for the 
    hole in $C$. In general, $\sjiparambisim{\aenv}$ is not preserved
    by contexts: 
    Consider processes $\aproc = a + b$, $\bproc = a$, environment
    $\aenv = \actpref{\aact}{\bact}$ and context $C = \actpref{\aact}{[]}$. We have that 
    $\aproc \sjiparambisim{\aenv} \bproc$, but $C[\aproc] = \actpref{\aact}{\aact + \bact} \not\sjiparambisim{\aenv} \actpref{\aact}{\aact} = C[\bproc]$. Notice
    that the above example also applies to $\sparambisim{\aenv}$.
    Indeed, including a process in a context intuitively also affects
    the environment, as shown in a study of the
%    In a sense, the context ``transforms the environment''.
    compositionality of $\sparambisim{\aenv}$ in \cite{lars:1987}. 
    The idea in that work is to introduce
    parametric environment-transformer\footnote{We use a 
    different notation than \cite{lars:1987}. There, 
    $T_C(\aenv)$ is rendered as $wie_{E\!\!\!\!E}(C,\aenv)$} $T_C$ 
    which preserves $\sparambisim{\aenv}$ in the following sense:
    \[
    \aproc\;\parambisim{T_C(\aenv)}\; \bproc \;\;\implies\;\; 
        <\!C,\aproc\!>\; \equiv_{\aenv}\; <\!C,\bproc\!> \punc{,}
    \]
    where $<\!C,\aproc\!>\; \equiv_{\aenv}\; <\!C,\bproc\!>$
    intuitively means that ``$C[\aproc] \jiparambisim{\aenv} 
    C[\bproc]$ 
    with $C$ interacting identically with $\aproc$ and $\bproc$'' 
    \cite{lars:1987} (we omit the formal definition for brevity).
    We speculate that, for $\sjiparambisim{\aenv}$, the requirement
    ``$C$ interacting identically with $\aproc$ and $\bproc$'' could
    be removed. If so, the compositionality of $\sjiparambisim{\aenv}$
    could be expressed as follows (for an appropriate environment-
    transformer $T'_C$):
    \[
    \aproc \jiparambisim{T'_C(\aenv)} \bproc \;\;\implies\;\; 
        C[\aproc] \jiparambisim{\aenv} C[\bproc] \punc{.}
    \]
%    Compositionality of $\sparambisim{\aenv}$ has been
%    studied in \cite{lars:1987}. The result obtained there is however
%    slighly different than our desiderata:
%    \[
%    \aproc\;\sparambisim{T_C(\aenv)}\; \bproc \implies 
%        <\!C,\aproc\!>\; \equiv_{\aenv}\; <\!C,\bproc\!>
%    \]
%    Where $<\!C,\aproc\!>\; \equiv_{\aenv}\; <\!C,\bproc\!>$
%    intuitively means that $C[\aproc] \sjiparambisim{\aenv} C[\bproc]$ 
%    with $C$ interacting identically with $\aproc$ and $\bproc$ 
%    \cite{lars:1987} (we omit the formal definition for brevity).
%    We speculate
%    A process can be ``applied'' to
%    a process to get
%    
%    By compositionality we mean 
%    the behaviour of $\sjiparambisim{\aenv}$ up to context, where 
%    as usual a context is a process with a hole that can be ``applied'' to
%    a process to get
%    
%    Larsen \cite{lars:1987} studied compositionality of 
%    $\sparambisim{\aenv}$ w.r.t. contextes. Contextes there are 
%    abstractly defined as a kind of process transformers, that is, 
%    functions that take a process and return another process. Such 
%    notions of context is shown equivalent to the standard notion
%    of a process with a (single) hole, where the transformation is just
%    the syntactic substitution of $\aproc$ for the hole. 
% \end{itemize}
% 
% 
% 
% \smallskip
% \noindent
% {\bf Related Work.}        
%   \todo{May be adapted and moved to the intro}
%  %        
%   %Related work: \todo{May be adapted and moved to the intro}
% \begin{itemize}
%   
\item{}\label{F2}
   The relatively recent work \cite{lago:murg:2023}  
shows that from $\sparambisim{\aenv}$ (and quantitative 
generalizations of it) one can extract a generalized pseudo-metric
between processes, where the codomain of the metric is the set of 
environments (under some closure assumptions). 
The idea is that the distance 
$d(\aproc,\bproc)$ between processes $\aproc,\bproc$ is defined
as the largest environment $\aenv$ (according to %eq. 
\eqref{eq:def:discrpo}) such that 
$\aproc \parambisim{\aenv} \bproc$. An obvious future work
is exploring whether a metric can be extracted for 
$\sjiparambisim{\aenv}$. The main challenge is finding the right
notion of ``largest environment'' for $\sjiparambisim{\aenv}$,
which is related to open problem (P1).  
\end{enumerate}

\smallskip
\noindent
{\bf Acknowledgment.}
  We thank the reviewers for their questions 
    about the literature on parameterized bisimilarity since its inception, %s in the 1980-ies.
    and concerning our work on the logical characterization of $\sjiparambisim{\aenv}$.
  We are also grateful for lists of comments that pointed us to necessary corrections and adaptations of details. 
  We thank Cl\'{e}ment Aubert for a direct discussion after the workshop, and in particular for pointing us to interesting connections
    with work \cite{abs-2007-08187,AubertV22} of his and Daniele Varacca.
  %We thank Clément Aubert for his comments on our work and for bringing his works \cite{abs-2007-08187,AubertV22} under our attention.
  

%% file: IntConcExp-2025-final.bbl
\begin{thebibliography}{10}
\providecommand{\bibitemdeclare}[2]{}
\providecommand{\surnamestart}{}
\providecommand{\surnameend}{}
\providecommand{\urlprefix}{Available at }
\providecommand{\url}[1]{\texttt{#1}}
\providecommand{\href}[2]{\texttt{#2}}
\providecommand{\urlalt}[2]{\href{#1}{#2}}
\providecommand{\doi}[1]{doi:\urlalt{https://doi.org/#1}{#1}}
\providecommand{\eprint}[1]{arXiv:\urlalt{https://arxiv.org/abs/#1}{#1}}
\providecommand{\bibinfo}[2]{#2}

\bibitemdeclare{inproceedings}{abs-2007-08187}
\bibitem{abs-2007-08187}
\bibinfo{author}{Cl{\'{e}}ment \surnamestart Aubert\surnameend} \&
  \bibinfo{author}{Daniele \surnamestart Varacca\surnameend}
  (\bibinfo{year}{2021}): \emph{\bibinfo{title}{Process, Systems and Tests:
  Three Layers in Concurrent Computation}}.
\newblock In: {\slshape \bibinfo{booktitle}{{ICE}}}, {\slshape
  \bibinfo{series}{{EPTCS}}} \bibinfo{volume}{347}, pp. \bibinfo{pages}{1--21},
  \doi{10.4204/EPTCS.347.1}.

\bibitemdeclare{article}{AubertV22}
\bibitem{AubertV22}
\bibinfo{author}{Cl{\'{e}}ment \surnamestart Aubert\surnameend} \&
  \bibinfo{author}{Daniele \surnamestart Varacca\surnameend}
  (\bibinfo{year}{2022}): \emph{\bibinfo{title}{Processes against tests: On
  defining contextual equivalences}}.
\newblock {\slshape \bibinfo{journal}{J. Log. Algebraic Methods Program.}}
  \bibinfo{volume}{129}, p. \bibinfo{pages}{100799},
  \doi{10.1016/J.JLAMP.2022.100799}.

\bibitemdeclare{inproceedings}{bene:kvre:lars:moll:srba:2011}
\bibitem{bene:kvre:lars:moll:srba:2011}
\bibinfo{author}{Nikola \surnamestart Bene{\v{s}}\surnameend},
  \bibinfo{author}{Jan \surnamestart K{\v{r}}et{\'i}nsk{\'y}\surnameend},
  \bibinfo{author}{Kim~G. \surnamestart Larsen\surnameend},
  \bibinfo{author}{Mikael~H. \surnamestart M{\o}ller\surnameend} \&
  \bibinfo{author}{Ji{\v{r}}{\'i} \surnamestart Srba\surnameend}
  (\bibinfo{year}{2011}): \emph{\bibinfo{title}{{Parametric Modal Transition
  Systems}}}.
\newblock In \bibinfo{editor}{Tevfik \surnamestart Bultan\surnameend} \&
  \bibinfo{editor}{Pao-Ann \surnamestart Hsiung\surnameend}, editors: {\slshape
  \bibinfo{booktitle}{Automated Technology for Verification and Analysis}},
  \bibinfo{publisher}{Springer Berlin Heidelberg}, \bibinfo{address}{Berlin,
  Heidelberg}, pp. \bibinfo{pages}{275--289},
  \doi{10.1007/978-3-642-24372-1_20}.

\bibitemdeclare{incollection}{berg:klop:1990}
\bibitem{berg:klop:1990}
\bibinfo{author}{J.~A. \surnamestart Bergstra\surnameend} \&
  \bibinfo{author}{J.~W. \surnamestart Klop\surnameend} (\bibinfo{year}{1990}):
  \emph{\bibinfo{title}{{An Introduction to Process Algebra}}}.
\newblock In \bibinfo{editor}{J.~C.~M. \surnamestart Baeten\surnameend},
  editor: {\slshape \bibinfo{booktitle}{{Applications of Process Algebra}}},
  \bibinfo{series}{Cambridge Tracts in Theoretical Computer Science},
  \bibinfo{publisher}{Cambridge University Press}, p. \bibinfo{pages}{1–22},
  \doi{10.1017/CBO9780511608841}.
\newblock
  \urlprefix\url{https://dspace.library.uu.nl/bitstream/handle/1874/13555/bergstra_90_introduction_process.pdf?sequence=3}.

\bibitemdeclare{techreport}{berg:klop:1986}
\bibitem{berg:klop:1986}
\bibinfo{author}{Jan~A \surnamestart Bergstra\surnameend} \&
  \bibinfo{author}{Jan~Willem \surnamestart Klop\surnameend}
  (\bibinfo{year}{1986}): \emph{\bibinfo{title}{{Algebra of Communicating
  Processes}}}.
\newblock \bibinfo{type}{Technical Report} \bibinfo{number}{89-138},
  \bibinfo{institution}{CWI Amsterdam}.
\newblock \urlprefix\url{https://ir.cwi.nl/pub/1778/1778D.pdf}.

\bibitemdeclare{inproceedings}{lago:murg:2023}
\bibitem{lago:murg:2023}
\bibinfo{author}{Ugo \surnamestart Dal~Lago\surnameend} \&
  \bibinfo{author}{Maurizio \surnamestart Murgia\surnameend}
  (\bibinfo{year}{2023}): \emph{\bibinfo{title}{{Contextual Behavioural
  Metrics}}}.
\newblock In \bibinfo{editor}{Guillermo~A. \surnamestart P\'{e}rez\surnameend}
  \& \bibinfo{editor}{Jean-Fran\c{c}ois \surnamestart Raskin\surnameend},
  editors: {\slshape \bibinfo{booktitle}{34th International Conference on
  Concurrency Theory (CONCUR 2023)}}, {\slshape \bibinfo{series}{Leibniz
  International Proceedings in Informatics (LIPIcs)}} \bibinfo{volume}{279},
  \bibinfo{publisher}{Schloss Dagstuhl -- Leibniz-Zentrum f{\"u}r Informatik},
  \bibinfo{address}{Dagstuhl, Germany}, pp. \bibinfo{pages}{38:1--38:17},
  \doi{10.4230/LIPIcs.CONCUR.2023.38}.

\bibitemdeclare{inproceedings}{fahr:lars:lega:trao:2014}
\bibitem{fahr:lars:lega:trao:2014}
\bibinfo{author}{Uli \surnamestart Fahrenberg\surnameend}, \bibinfo{author}{Kim
  \surnamestart Guldstrand~Larsen\surnameend}, \bibinfo{author}{Axel
  \surnamestart Legay\surnameend} \& \bibinfo{author}{Louis-Marie \surnamestart
  Traonouez\surnameend} (\bibinfo{year}{2014}):
  \emph{\bibinfo{title}{{Parametric and Quantitative Extensions of Modal
  Transition Systems}}}.
\newblock In \bibinfo{editor}{Saddek \surnamestart Bensalem\surnameend},
  \bibinfo{editor}{Yassine \surnamestart Lakhneck\surnameend} \&
  \bibinfo{editor}{Axel \surnamestart Legay\surnameend}, editors: {\slshape
  \bibinfo{booktitle}{From Programs to Systems. The Systems perspective in
  Computing}}, \bibinfo{publisher}{Springer Berlin Heidelberg},
  \bibinfo{address}{Berlin, Heidelberg}, pp. \bibinfo{pages}{84--97},
  \doi{10.1007/978-3-642-54848-2_6}.

\bibitemdeclare{article}{henn:miln:1985}
\bibitem{henn:miln:1985}
\bibinfo{author}{Matthew \surnamestart Hennessy\surnameend} \&
  \bibinfo{author}{Robin \surnamestart Milner\surnameend}
  (\bibinfo{year}{1985}): \emph{\bibinfo{title}{{Algebraic Laws for
  Nondeterminism and Concurrency}}}.
\newblock {\slshape \bibinfo{journal}{J. ACM}}
  \bibinfo{volume}{32}(\bibinfo{number}{1}), p. \bibinfo{pages}{137–161},
  \doi{10.1145/2455.2460}.

\bibitemdeclare{inproceedings}{huth:jaga:schm:2001}
\bibitem{huth:jaga:schm:2001}
\bibinfo{author}{Michael \surnamestart Huth\surnameend}, \bibinfo{author}{Radha
  \surnamestart Jagadeesan\surnameend} \& \bibinfo{author}{David \surnamestart
  Schmidt\surnameend} (\bibinfo{year}{2001}): \emph{\bibinfo{title}{{Modal
  Transition Systems: A Foundation for Three-Valued Program Analysis}}}.
\newblock In \bibinfo{editor}{David \surnamestart Sands\surnameend}, editor:
  {\slshape \bibinfo{booktitle}{Programming Languages and Systems}},
  \bibinfo{publisher}{Springer Berlin Heidelberg}, \bibinfo{address}{Berlin,
  Heidelberg}, pp. \bibinfo{pages}{155--169}, \doi{10.1007/3-540-45309-1_11}.

\bibitemdeclare{misc}{lago:murg:2023:arxiv}
\bibitem{lago:murg:2023:arxiv}
\bibinfo{author}{Ugo~Dal \surnamestart Lago\surnameend} \&
  \bibinfo{author}{Maurizio \surnamestart Murgia\surnameend}
  (\bibinfo{year}{2023}): \emph{\bibinfo{title}{{Contextual Behavioural Metrics
  (Extended Version)}}}, \doi{10.48550/arXiv.2307.07400}.
\newblock \eprint{2307.07400}.

\bibitemdeclare{phdthesis}{lars:1986}
\bibitem{lars:1986}
\bibinfo{author}{Kim~G. \surnamestart Larsen\surnameend}
  (\bibinfo{year}{1986}): \emph{\bibinfo{title}{{Context-Dependent Bisimulation
  between Processes}}}.
\newblock Ph.D. thesis, \bibinfo{school}{University of Edinburgh}.
\newblock
  \urlprefix\url{https://era.ed.ac.uk/bitstream/handle/1842/11030/Larsen1986.pdf}.

\bibitemdeclare{article}{lars:1987}
\bibitem{lars:1987}
\bibinfo{author}{Kim~G. \surnamestart Larsen\surnameend}
  (\bibinfo{year}{1987}): \emph{\bibinfo{title}{{A Context Dependent
  Equivalence between Processes}}}.
\newblock {\slshape \bibinfo{journal}{Theoretical Computer Science}}
  \bibinfo{volume}{49}(\bibinfo{number}{2}), pp. \bibinfo{pages}{185--215},
  \doi{10.1016/0304-3975(87)90007-7}.

\bibitemdeclare{inproceedings}{lars:lars:waso:2005}
\bibitem{lars:lars:waso:2005}
\bibinfo{author}{Kim~G. \surnamestart Larsen\surnameend},
  \bibinfo{author}{Ulrik \surnamestart Larsen\surnameend} \&
  \bibinfo{author}{Andrzej \surnamestart Wasowski\surnameend}
  (\bibinfo{year}{2005}): \emph{\bibinfo{title}{Color-Blind Specifications for
  Transformations of Reactive Synchronous Programs}}.
\newblock In \bibinfo{editor}{Maura \surnamestart Cerioli\surnameend}, editor:
  {\slshape \bibinfo{booktitle}{Fundamental Approaches to Software
  Engineering}}, \bibinfo{publisher}{Springer Berlin Heidelberg},
  \bibinfo{address}{Berlin, Heidelberg}, pp. \bibinfo{pages}{160--174},
  \doi{10.1007/978-3-540-31984-9_13}.

\bibitemdeclare{article}{lars:miln:1992}
\bibitem{lars:miln:1992}
\bibinfo{author}{Kim~G. \surnamestart Larsen\surnameend} \&
  \bibinfo{author}{Robin \surnamestart Milner\surnameend}
  (\bibinfo{year}{1992}): \emph{\bibinfo{title}{{A Compositional Protocol
  Verification Using Relativized Bisimulation}}}.
\newblock {\slshape \bibinfo{journal}{Information and Computation}}
  \bibinfo{volume}{99}(\bibinfo{number}{1}), pp. \bibinfo{pages}{80--108},
  \doi{10.1016/0890-5401(92)90025-B}.

\bibitemdeclare{inproceedings}{lars:nyma:wkas:2007}
\bibitem{lars:nyma:wkas:2007}
\bibinfo{author}{Kim~G. \surnamestart Larsen\surnameend},
  \bibinfo{author}{Ulrik \surnamestart Nyman\surnameend} \&
  \bibinfo{author}{Andrzej \surnamestart W{\k{a}}sowski\surnameend}
  (\bibinfo{year}{2007}): \emph{\bibinfo{title}{{On Modal Refinement and
  Consistency}}}.
\newblock In \bibinfo{editor}{Lu{\'i}s \surnamestart Caires\surnameend} \&
  \bibinfo{editor}{Vasco~T. \surnamestart Vasconcelos\surnameend}, editors:
  {\slshape \bibinfo{booktitle}{CONCUR 2007 -- Concurrency Theory}},
  \bibinfo{publisher}{Springer Berlin Heidelberg}, \bibinfo{address}{Berlin,
  Heidelberg}, pp. \bibinfo{pages}{105--119},
  \doi{10.1007/978-3-540-74407-8_8}.

\bibitemdeclare{inproceedings}{lars:1990}
\bibitem{lars:1990}
\bibinfo{author}{Kim~Guldstrand \surnamestart Larsen\surnameend}
  (\bibinfo{year}{1990}): \emph{\bibinfo{title}{{Modal Specifications}}}.
\newblock In \bibinfo{editor}{Joseph \surnamestart Sifakis\surnameend}, editor:
  {\slshape \bibinfo{booktitle}{Automatic Verification Methods for Finite State
  Systems}}, \bibinfo{publisher}{Springer Berlin Heidelberg},
  \bibinfo{address}{Berlin, Heidelberg}, pp. \bibinfo{pages}{232--246},
  \doi{10.1007/3-540-52148-8_19}.

\bibitemdeclare{book}{miln:1980}
\bibitem{miln:1980}
\bibinfo{author}{Robin \surnamestart Milner\surnameend} (\bibinfo{year}{1980}):
  \emph{\bibinfo{title}{{A Calculus of Communicating Systems}}},
  \bibinfo{edition}{1980} edition.
\newblock \bibinfo{publisher}{Springer Berlin Heidelberg},
  \doi{10.1007/3-540-10235-3}.

\bibitemdeclare{inproceedings}{miln:1985}
\bibitem{miln:1985}
\bibinfo{author}{Robin \surnamestart Milner\surnameend} (\bibinfo{year}{1985}):
  \emph{\bibinfo{title}{{Lectures on a Calculus for Communicating Systems}}}.
\newblock In \bibinfo{editor}{Stephen~D. \surnamestart Brookes\surnameend},
  \bibinfo{editor}{Andrew~William \surnamestart Roscoe\surnameend} \&
  \bibinfo{editor}{Glynn \surnamestart Winskel\surnameend}, editors: {\slshape
  \bibinfo{booktitle}{Seminar on Concurrency}}, \bibinfo{publisher}{Springer
  Berlin Heidelberg}, \bibinfo{address}{Berlin, Heidelberg}, pp.
  \bibinfo{pages}{197--220}, \doi{10.1007/3-540-15670-4_10}.

\bibitemdeclare{article}{pier:sang:2000}
\bibitem{pier:sang:2000}
\bibinfo{author}{Benjamin~C. \surnamestart Pierce\surnameend} \&
  \bibinfo{author}{Davide \surnamestart Sangiorgi\surnameend}
  (\bibinfo{year}{2000}): \emph{\bibinfo{title}{{Behavioral Equivalence in the
  Polymorphic Pi-Calculus}}}.
\newblock {\slshape \bibinfo{journal}{J. ACM}}
  \bibinfo{volume}{47}(\bibinfo{number}{3}), p. \bibinfo{pages}{531–584},
  \doi{10.1145/337244.337261}.

\bibitemdeclare{book}{sang:2011}
\bibitem{sang:2011}
\bibinfo{author}{Davide \surnamestart Sangiorgi\surnameend}
  (\bibinfo{year}{2011}): \emph{\bibinfo{title}{Introduction to Bisimulation
  and Coinduction}}.
\newblock \bibinfo{publisher}{Cambridge University Press},
  \doi{10.1017/CBO9780511777110}.

\end{thebibliography}
